\documentclass[12pt,reqno]{amsart}

\usepackage{amscd,amssymb,amsmath,latexsym}
\usepackage[mathscr]{euscript}
\usepackage{epsfig}
\usepackage{xcolor} 
\usepackage{bm} 
\usepackage{hyperref} 
\allowdisplaybreaks 
\usepackage{physics}
\usepackage[left=.5in,right=.5in,top=1in,bottom=1in]{geometry}
\usepackage{verbatim}
\usepackage{multirow} 
\usepackage{float} 

\usepackage{cancel} 

\newcommand{\N}{\mathbb N}

\newcommand{\C}{\mathbb C}

\renewcommand{\P}{\mathbb P}

\allowdisplaybreaks

\makeatletter
\def\blfootnote{\xdef\@thefnmark{}\@footnotetext}
\makeatother

\newtheorem{Thm}{Theorem}[section]
\newtheorem{Cor}[Thm]{Corollary}
\newtheorem{Lem}[Thm]{Lemma}

\newtheorem{Rmk}[Thm]{Remark}

\newtheorem{Def}[Thm]{Definition}
\newtheorem{Conv}[Thm]{Convention}
\begin{document}

\title{Optimal lower bound for lossless quantum block encoding}

\author{George Androulakis, Rabins Wosti}\blfootnote{This article is part of the PhD dissertation of the second author who is a graduate student in the department of Computer Science and Engineering of the University of South Carolina. A part of this work will be presented in July at 2024 IEEE International Symposium on Information Theory (ISIT 2024).}

\address{Department of Mathematics, University of South Carolina, 
Columbia, SC 29208}
\email{giorgis@math.sc.edu}

\address{Department of Computer Science and Engineering, University of South Carolina, Columbia, SC 29208}

\email{rwosti@email.sc.edu}

\begin{abstract}
Consider a general quantum stochastic source that emits at discrete 
time steps quantum pure states which are chosen from a finite alphabet according to some probability distribution 
which may depend on the whole history.  Also, fix two 
positive integers $m$ and $l$.
We encode any
tensor product of $ml$ many states emitted by the quantum stochastic source by 
breaking the tensor product into $m$ many blocks where each block has 
length $l$, and considering 
sequences of $m$ many isometries so that each isometry
encodes one of these blocks into the Fock space, followed by  
the concatenation of their images. We only consider certain 
sequences of such isometries that we call 
\lq\lq special block codes\rq\rq\
in order to ensure that the string of encoded states is uniquely decodable. We compute the minimum average codeword length
of these encodings which depends on the quantum source
and the integers $m$, $l$, among all possible special block codes.
Our result extends the result of
[Bellomo, Bosyk, Holik and Zozor,  Scientific Reports 7.1 (2017): 14765] where the minimum was computed for one block,
i.e.\ for $m=1$. Lastly, we give a simplified non-adaptive compression technique based on constrained special block codes for general quantum stochastic sources. For quantum stationary sources in particular, we show that the minimum average codeword length per symbol computed over all constrained special block codes is equal to the von-Neumann entropy rate of the source for an asymptotically long block size.
\end{abstract}

\keywords{Indeterminate length, Kraft-McMillan Inequality, uniquely decodable code, Fock space}
\subjclass[2020]{Primary: 81P45, Secondary: 94A15.}
\maketitle 

\section{Introduction}
Important foundations to the area of quantum encoding were provided by Schumacher 
who proved the quantum analog of 
Shannon’s noiseless coding theorem for an independent and 
identically distributed (iid) quantum source, \cite{Schumacher1995Apr}. 
In this article we study lossless, variable-length block encoding of quantum information
which is emitted from a completely general quantum source, and it is encoded into the Fock space. 
Since the Schumacher's mentioned result, 
various block encoding schemes for quantum sources have been 
proposed. We review these schemes in the next few paragraphs of the current section. At the end of the section we 
describe our contribution in more detail. 

Chuang and Modha \cite{Chuang2000May} proposed a 
source-dependent quantum algorithm to project the block quantum 
state of an iid quantum source to its \lq\lq typical subspace\rq\rq. The projected state is compressed using the quantum 
analog of Shannon-Fano code. According to that code,  
the number of qubits per 
signal is only slightly larger than the von~Neumann entropy of the 
source for sufficiently large block size. 

A universal fixed-length quantum 
block compression scheme was introduced by Jozsa et. al. \cite{Jozsa1998Aug} for a family of iid quantum sources such that the 
von~Neumann entropy of each source in the family is known to be upper bounded by a fixed number $S$. Their encoding is source-independent, in the sense that one does not need to know 
the average ensemble state which is emitted by the quantum  
source. In particular, they 
showed that for a sufficiently large block size, their universal encoding compresses the quantum information from any of these sources to $S$ qubits/signal and is asymptotically optimal. Kalchenko and Yang extended this result in \cite{Kaltchenko2003Aug, Kaltchenko2003}. They proved that 
for a family of stationary ergodic quantum sources such that the von~Neumann entropy rate of each source in the family is 
upper bounded by a real number $r > 0$, there exists a 
sequence of universal projectors which 
compresses the quantum information from any of these sources to $r$ qubits/signal with high fidelity for sufficiently large block 
size. Universal variable-length lossless quantum data compression schemes have been discussed by Hayashi and Matsumoto in \cite{Hayashi2002Aug, HayashiMasahito2002Dec} and by Hayashi in \cite{Hayashi2010Jan}.

Another source-dependent quantum block compression was 
introduced by Szeto \cite{Szeto1996Jul} for an 
iid quantum source. This encoding achieves the 
same compression performance as Schumacher’s noiseless coding theorem \cite{Schumacher1995Apr}, 
but employs a classical block compression alongside, and
adopts
the assumption that the Hilbert space can be decomposed into two 
or more mutually orthogonal subspaces and each signal from the 
source comes from one of these subspaces. 

Langford \cite{Langford2002May} gave a compilation algorithm that takes two inputs. The first input is any classical block compression scheme that uses bounded space and time,
and the second input is the density matrix of an iid quantum source. The algorithm outputs a fixed-length quantum block compression algorithm. If the input classical algorithm asymptotically approaches the Shannon entropy of the classical source, then the output quantum block compression algorithm also asymptotically 
approaches the von~Neumann entropy of the input quantum source.



There are several completely lossless 
compression algorithms in the classical information theory, 
such as Huffman encoding \cite{Huffman1952Sep}, Shannon-Fano encoding \cite{shannon, fano}and 
arithmetic encoding \cite{arithmetic}. It 
is natural to attempt to generalize these classical 
algorithms to the quantum realm in order to obtain lossless 
quantum 
data compression. Classical completely lossless block 
compression algorithms make use of variable-length codes 
where sequences of symbols with high probability of 
occurrence are assigned codewords of shorter length and less 
probable sequences have longer codewords, thereby reducing 
the average codeword length. This idea extends to 
quantum data compression. In the quantum setting, 
not only the codewords can have different lengths,
but also the codewords can be in superposition of states of 
different length and they are called indeterminate-length 
codewords. Boström and Felbinger \cite{Bostroem2002Feb} show that it is impossible to achieve a 
completely lossless quantum block compression scheme by 
using fixed-length codewords. Therefore, the
indeterminate-length codewords form an integral part of 
lossless quantum data compression, and were investigated first by Schumacher \cite{Schumacher1994} and later was formalized by Schumacher and Westmoreland \cite{Schumacher2001Sep}.

The quantum analogue of classical Huffman code was investigated by Braunstein, Fuchs, Gottesman and Lo \cite{Braunstein2000Jul}. Schumacher and Westmoreland \cite{Schumacher2001Sep} proved that for an iid source,
the quantum Huffman code is optimal
in the sense that it minimizes the average codeword length which asymptotically approaches the von~Neumann entropy of 
the source as the block size increases to infinity.
For their proof they used the formalism of condensable 
codes. Bellomo, Bosyk, Holik and Zozor \cite[Theorem 2]{Bellomo2017Nov}
proved later that the quantum Huffman code minimizes the 
average codeword length of a single block for any given 
quantum source, (not necessarily iid). There has also been some 
parallel works \cite{Muller2008Jan, Chou2006} towards minimizing the average base length 
(which is the maximum length of an indeterminate-length codeword used) for quantum sources as opposed to minimizing the average codeword length. 

During the recent years, there have been many developments 
in quantum communication and data compression. We refer the 
reader to the recent textbooks \cite{Khatri2020Nov} by Khatri and 
Wilde and \cite{Hajdusek2023Nov} by Hajdu\u{s}ek and Van Meter 
for a comprehensive review on the subject.

In classical (or quantum) data transmission, the number of bits (or qubits) required to transmit any given classical (or quantum) message is considered as a resource. So, it is desirable to transmit a given classical (or quantum) message using as few bits (or qubits) as possible. Therefore, in this paper, we focus on quantum block data compression in an effort to minimize the average codeword length for a given length of the quantum message emitted by a given quantum source. Precisely, we study \textit{completely lossless} and \textit{indeterminate-length} quantum block data compression of pure states emitted by a \textit{quantum stochastic source}, which is not necessarily iid. We will follow the formalism in \cite{Bellomo2017Nov} as opposed to the zero-extended forms of \cite{Schumacher2001Sep} in order to represent indeterminate-length codewords.

Consider a \textit{quantum stochastic source} (not necessarily iid) and a tensor product of $q \in \N$ many pure states emitted from the source. Fix $m \in \N$ as the number of blocks and $l$ as the block size such that $q=ml$. Consider a sequence of $m$ many isometries (also called quantum codes) $U_1, \ldots, U_m$ such that for $k \in \{1, \ldots, m\}$, each $U_k$ encodes the $k^{th}$ block consisting of the tensor product of $l$ many pure states into the Fock space. The given sequence of $m$ isometries thus encodes the tensor product of $ml$ many states to the state obtained by the concatenation of the images of all $U_k$'s. In order to decode the encoded state in a lossless manner, the sequence of quantum codes must be ``uniquely decodable'' (also called completely lossless) quantum code. Additionally, this kind of approach of block encoding of the pure states raises some interesting questions: for a given message consisting of the tensor product of $q = ml$ many pure states emitted from a given quantum stochastic source $\mathcal{S}$, what is the minimum average codeword length per symbol computed over all uniquely decodable quantum codes $U_1, \ldots, U_m$ such that each $U_k$ encodes the $k^{th}$ block of $l$ many pure states? How does the minimum average codeword length per symbol computed over all uniquely decodable quantum codes change if we change $m$ and $l$ such that their product $ml$ always remains constant?

To this end, in this paper we extend the notion of uniquely decodable (or completely lossless) quantum codes to be used for quantum block data compression. As the main result (Theorem~\ref{Thm:main}), for a fixed $ml$ many pure states emitted by a given quantum stochastic source, we derive the optimal lower bound of the average codeword length over a subset of uniquely decodable quantum codes called ``special block codes", which are applied to encode the pure states in $m$ many blocks each of block size $l$. We hope that our work would be supportive in the quest towards answering the above questions, which have been left open in the current work. 

We explicitly show that in order to achieve the optimal lower bound, for every $k \in \{1, \ldots, m\}$, one needs to dynamically adapt the isometry used for encoding the $k^{th}$ block of $l$ many pure states based on the exact sequence of pure states emitted in the first $(k-1)$ blocks. However, since the pure states emitted by the source are not necessarily orthogonal, the task of state discrimination of non-orthogonal states can be difficult. So, in Section~\ref{Stationary}, we give a simplified non-adaptive compression technique based on a subset of special block codes called ``constrained special block codes" for quantum stochastic sources that emit pure states where the encoder can fix one isometry per block of states ignoring the emissions in the previous blocks, and in Theorem~\ref{Thm:derived} we derive the optimal lower bound of the average codeword length over all constrained special block codes which are applied to encode the pure states in $m$ many blocks each of block size $l$. In the Corollary~\ref{corollary_stationary}, we show that for quantum stationary sources in particular, the optimal lower bound of the average codeword length per symbol computed over all constrained special block codes equals the von-Neumann entropy rate of the source for an asymptotically long block size.

\section{Preliminary definitions} 

A classical binary encoding for a finite alphabet set $A$ is a map from the set $A$ to the set of all finite strings
of $0$'s and $1$'s. The images of the elements of the alphabet via an encoding are called codewords. 
Every encoding of an alphabet set $A$ can be extended to the set of finite strings of elements of $A$
by concatenation. An encoding is called uniquely decodable if its extension is one-to-one. An encoding
is called instantaneous (or prefix-free) if no codeword is the initial part of another codeword. Instantaneous encodings are
special cases of uniquely decodable ones. If the 
information source is independent and identically distributed (iid), then a probability distribution is assigned to the alphabet,
and the average length of the code is defined to be equal to the expectation of the length of its codewords, where the 
length of each codeword is the number of its bits. The Huffman encoding \cite{Huffman1952Sep} is an 
optimal instantaneous classical encoding and its average codeword length $l_{avg}$ satisfies
$$
H(X) \leq l_{avg} \leq H(X) + 1  
$$
where $H(X)$ is Shannon entropy of a classical source.

Now we look at the analogous situation of the quantum data compression.
Throughout this article, we fix a Hilbert space $\mathcal{H}$ of dimension $d$ and a set  $\{\ket{s_n}\}_{n=1}^N$  of pure 
states of  $\mathcal{H}$. We assume without loss of generality that the vectors $\{\ket{s_n}\}_{n=1}^N$  span 
the Hilbert space $\mathcal{H}$.
In this section we recall the definitions of  quantum stochastic sources, quantum stochastic ensembles, quantum codes and average codeword length which have been considered in some of the references mentioned above.

\begin{Def}\label{quantum_source}
A \textbf{quantum stochastic source} $\mathcal{S}$ consists of a set of 
pure states $\{\ket{s_n}\}_{n=1}^N$ of a 
Hilbert space $\mathcal{H}$, and a stochastic process 
$X=(X_n)_{n=1}^\infty$, where each $X_n$ is a random 
variable which takes values in $\{ 1, 2, \ldots , N \}$. 
At every positive integer time $n$ the state $\ket{s_{X_n}}$
is emitted from the quantum source. If $p$ denotes the 
probability distribution of the stochastic process $X$, then
for every $q \in \N$ and 
$(n_1, \ldots, n_q) \in \{ 1, \ldots , N \}^q$, we have that 
\begin{equation} \label{stochastic_px}
p(n_1, \ldots , n_q)= \mathbb{P}(X_1=n_1, \ldots , X_q=n_q).
\end{equation}
Also, the conditional probability distribution of the stochastic process $X$ is defined for any $q \geq 2$ and 
$(n_1, \ldots, n_q) \in \{ 1, \ldots ,N \}^q$ by
$$
p(n_q|n_{q-1}, \ldots ,n_1)= \P(X_q=n_q|X_{q-1}=n_{q-1}, \ldots ,X_1=n_1).
$$
\end{Def}

\begin{Def} \label{stationary_source}
A quantum stochastic source $\mathcal{S}$ is \textbf{stationary (or translation-invariant)} if the associated stochastic process $X$ is invariant with respect to the translation map, i.e. 
\begin{equation} \label{translation-invariance}
\mathbb{P}(X_1=n_1, \ldots , X_q=n_q) = \mathbb{P}(X_{k+1}=n_1, \ldots , X_{k+q}=n_q)
\end{equation}
for every $k \in \N$, 
$q \in \N$ and 
$(n_1, \ldots, n_q) \in \{ 1, \ldots , N \}^q$.
\end{Def}
Setting $q=1$ in Equation~\eqref{translation-invariance}, we get 
$
\mathbb{P}(X_1=n_1) = \mathbb{P}(X_{k+1}=n_1)
$
for every $k \in \N$ and 
$n_1 \in \{ 1, \ldots , N \}$, which shows that for every quantum stationary source, the random variables are identically distributed. Notice however that the random variables are not necessarily independent. We will deal with quantum stationary sources in Section~\ref{Stationary}.
\begin {Def} \label{general_ensemble}
Given a quantum stochastic source $\mathcal{S}$ which consists of a set of 
pure states $\{\ket{s_n}\}_{n=1}^N$ of a 
Hilbert space $\mathcal{H}$, and a stochastic process 
$X=(X_n)_{n=1}^\infty$ with probability distribution $p$,
we define the \textbf{ensemble state} $\rho_q$ for $q \in \N$ many emissions from the source as follows: 
$$
\rho_q=\sum_{n_1,\ldots,n_q=1}^N p(n_1,\ldots,n_q)\ketbra{s_{n_1}\cdots s_{n_q}}{s_{n_1}\cdots s_{n_q}} .
$$
\end{Def}

In this article, we are concerned with quantum block encodings 
and we assume that each block that we encode is a tensor 
product of $l$ many pure states which are chosen from the set 
$\{ \ket{s_n}: n =1, \ldots , N \}$ via the 
stochastic process $(X_n)_{n \in \mathbb{N}}$. Usually we 
express that fact by saying that \textbf{the block size} is equal to $l$. Thus, we often
look at multiples of $l$ many pure states. This can be 
observed in the next definition.

\begin{Def} \label{block_ensemble_state}
Let a quantum stochastic source $\mathcal{S}$ which consists of a set of 
pure states $\{\ket{s_n}\}_{n=1}^N$ of a 
Hilbert space $\mathcal{H}$, and a stochastic process 
$X=(X_n)_{n=1}^\infty$ with probability distribution $p$.
For a fixed block size $l \in \N$, an integer $k \geq 1$, and a fixed sequence 
$(n_1, \ldots , n_{(k-1)l}) \in \{ 1, \ldots , N\}^{(k-1)l}$
define the 
$\boldsymbol{k}^{\boldsymbol{th}}$ \textbf{block conditional ensemble state} $\rho^{n_1, \ldots ,n_{(k-1)l}}$ as
\begin{align} 
\rho^{n_1, \ldots ,n_{(k-1)l}} & = \sum_{n_{(k-1)l+1},\ldots, n_{kl}=1}^N \frac{p(n_1,\ldots,n_{kl})}{p(n_1,\ldots,n_{(k-1)l})} \ketbra{s_{n_{(k-1)l+1}} \cdots s_{n_{kl}}}  \nonumber \\ 
& = \sum_{n_{(k-1)l+1},\ldots, n_{kl}=1}^N  p(n_{(k-1)l + 1},\ldots,n_{kl} | n_1,\ldots,n_{(k-1)l})  \ketbra{s_{n_{(k-1)l+1}} \cdots s_{n_{kl}}},\label{kth_block_ensemble}
\end{align}
where each coefficient $p(n_{(k-1)l + 1},\ldots,n_{kl} | n_1,\ldots,n_{(k-1)l})$ above represents the conditional probability that the sequence of states $\ket{s_{n_{(k-1)l + 1}}, \ldots, s_{n_{kl}}}$ gets emitted in the $k^{th}$ block given the emission of the sequence of states $\ket{s_{n_1}, \ldots, s_{n_{(k-1)l}}}$ in the first $(k-1)$ many blocks.
Consequently, the $\boldsymbol{k}^{\boldsymbol{th}}$ \textbf{block ensemble state} $\rho^k$ is given by
\begin{equation}
\rho^k = \sum_{n_1, \ldots, n_{(k-1)l} =1}^N p(n_1, \ldots, n_{(k-1)l}) \rho^{n_1, \ldots, n_{(k-1)l}}  \label{block_ensemble}
\end{equation}
\end{Def}
In what follows, for $k=1$, the state 
$\rho^{n_1, \ldots , n_{(k-1)l}}$ is interpreted as the 
ensemble state $\rho_\ell$ (or equivalently, $\rho^1$) of the first block, $p(n_1, \ldots, n_{(k-1)l})$ is set to 1, and the string $n_1, \ldots , n_{(k-1)l}$ is the empty string and is dropped from the superscript 
wherever applicable.

Recall that \textbf{Fock space} $\bm{(\mathbb{C}^2)^{\oplus}}$ is
associated to $\C^2$ and is defined to be the Hilbert space
\begin{equation}
(\mathbb{C}^2)^{\oplus} = \bigoplus_{k=0}^\infty (\mathbb{C}^2)^{\otimes k},
\end{equation}
where $(\mathbb{C}^2)^{\otimes 0}= \mathbb{C}$. We simply 
refer to the Fock space associated to $\mathbb{C}^2$ as the 
\textbf{Fock space}. The 
\textbf{standard basis} of the Fock space 
$(\mathbb{C}^2)^\oplus$ consists of qubit strings of various 
lengths, (including the $0$-length $\ket{\emptyset}$). 
If a quantum state $\ket{\psi}$ includes the superposition of only the qubit strings of some length $\ell \in \N$, i.e. $\ket{\psi} \in (\mathbb{C}^2)^{\otimes \ell}$, we refer to such state as \textbf{length state}.

Now, we define the notion of quantum codes.

\begin{Def} \label{def:quantum_code}
Let 
$\mathcal{K}$ be a Hilbert space of dimension $D$. 
A \textbf{quantum code} 
on $\mathcal{K}$ is a linear isometry $U:\mathcal{K} \to (\mathbb{C}^2)^\oplus$.
Thus, for every quantum code $U$ on $\mathcal{K}$, the dimension of its range is equal to $D$, and if $(\ket{\psi_i})_{i=1}^D$ is any orthonormal sequence in its range, then $U$ has the form
\begin {equation} \label{quantum_code}
U=\sum_{i=1}^{D} \ketbra{\psi_i}{e_i},
\end {equation}
where $(\ket{e_i})_{i=1}^D$ is an orthonormal basis of 
$\mathcal{K}$. The quantum state of the form $U\ket{s}$ obtained by applying $U$ to any pure state $\ket{s} \in \mathcal{K}$ is called \textbf{codeword}.
In particular, if the codeword is a length state, it is called \textbf{length codeword}. However, in general, a codeword can be a superposition of lengths states of different lengths, in which case the codeword is called \textbf{indeterminate-length codeword}.
\end{Def}
As a result, the range of a quantum code is the finite dimensional subspace of the Fock space that is spanned by the orthonormal vectors $\ket{\psi_i}$'s, possibly of different lengths. The notion of concatenation of codewords is defined as follows:
 
\begin{Def} \label{concatenation}
\textbf{Concatenation} of standard basis vectors of the Fock space is defined as the 
concatenation of classical bit strings. Then, \textbf{concatenation}
of the linear combinations of basis vectors of the Fock space is defined by making the concatenation distributive over addition and scalar multiplication, as 
in \cite[Definition 2.4]{Muller2008}  and  
\cite[Section V]{Muller2008Jan}.
Hence, for $m \in \N$ and $x_1, \ldots , x_m \in (\mathbb{C}^2)^\oplus$, their concatenation $x_1 \circ x_2 \circ \cdots \circ x_m$ is an element of  $(\mathbb{C}^2)^\oplus$.
Let $\mathcal{K}$ be a finite dimensional Hilbert space, 
$m \in \N$, and 
$U_1,  \ldots , U_m$ be a sequence of quantum codes with
$U_k: \mathcal{K} \to (\mathbb{C}^2)^\oplus$ for all 
$k =1, \ldots , m$. Then, the \textbf{concatenation} of the 
sequence $(U_k)_{k=1}^m$ is defined to be the 
linear map $U_1 \circ \cdots \circ U_m :\mathcal{K}^{\otimes m} \to (\mathbb{C}^2)^\oplus$ 
given by
\begin{equation} \label{extended_quantum}
(U_1 \circ \cdots \circ U_m)\ket{s_1 \otimes \ldots \otimes s_m}=U_1\ket{s_1} \circ  \cdots \circ U_m\ket{s_m} \text{ for all pure states}  \ \ket{s_1s_2\ldots s_m} \in \mathcal{K}^{\otimes m} ,
\end{equation}
where $\circ$ denotes the concatenation operation. 
\end{Def}

\begin{Rmk}
It is important to note that concatenation of arbitrary pure quantum states does not always preserve the norm, and hence may not be well defined. It's not hard to see that the norm is always preserved if either or both of the pure quantum states being concatenated are length states. 
However, if both pure quantum states being concatenated are indeterminate-length states, then the norm may not be preserved (for example, concatenation of the pure states $\frac{1}{\sqrt 2} (\ket{0} + \ket{00})$ and $\frac{1}{\sqrt 2} (\ket{0} - \ket{00}$ is not normalized). This has been well explored in \cite[after Definition 2.4]{Muller2008}.
On the other hand, Remark~\ref{rmk:well-definition_of_concatenation} will give that on certain subspaces of the Fock space, 
concatenations of fixed number of vectors are well defined.
\end{Rmk}

Thus,
if $(U_k)_{k=1}^m$ is a finite sequence of quantum codes, then $U_1 \circ  \cdots \circ U_m$ in general may not be an 
isometry. It is important that the concatenation
$U_1 \circ \cdots \circ U_m$ remains an isometry, because if $U_k$ is used to encode the $k^{\text{th}}$ block 
for each $k =1, \ldots , m$, then 
$U_1 \circ \cdots \circ U_m$ will be 
used to encode the tensor product of $m$ many blocks 
of pure states, so its inverse will be used to decode those 
blocks back in a lossless fashion. This lossless way of 
decoding is possible if the concatenation 
$U_1 \circ \cdots \circ U_m $ is uniquely decodable which is explained below.

We now extend the notion, in some sense, of uniquely decodable quantum codes considered in \cite{Hayashi2010Jan, Hayashi, Bellomo2017Nov, Androulakis2019Sep} from one isometry to a sequence of isometries.

\begin{Def} \label{uniquely_decodable}
Let  $\mathcal{K}$ be a finite dimensional Hilbert space, and 
$U_1, \ldots , U_m$ be a finite sequence of quantum codes such 
that $U_k:\mathcal{K} \to (\mathbb{C}^2)^\oplus$ for every 
$k =1, \ldots , m$. The sequence 
of quantum codes $(U_k)_{k=1}^m$ is called \textbf{uniquely decodable} if and only if their concatenation
$U_1 \circ \cdots \circ U_m: \mathcal{K}^{\otimes m} \to (\mathbb{C}^2)^\oplus$ is an isometry. 
\end{Def}

We now give a characterization of uniquely 
decodable quantum codes that follows directly from the 
Definitions~\ref{concatenation} and 
\ref{uniquely_decodable}.

\begin{Thm} \label{m-orthonormal}
Let $\mathcal{K}$ be a $D$-dimensional Hilbert space and 
$(U_k)_{k=1}^m$ be a sequence of isometries 
$U_k: \mathcal{K} \to (C^2)^{\oplus}$ with 
$U_k = \sum_{i=1}^D \ketbra{\psi_i^k}{e_i^k}$ (as in Equation~\eqref{quantum_code}) for $k=1,\ldots, m$.
Then, $(U_k)_{k=1}^m$ is uniquely decodable if and only if 
the collection
$$
\Big\{ \ket{\psi_{v_1}^1 \circ \cdots \circ \psi_{v_m}^m} : 
( v_1, \ldots , v_m ) \in \{ 1, \ldots , D\}^m
\Big\}
$$
forms an orthonormal set.
\end{Thm}
Proof: See Appendix \ref{proof_uniquely_decodable}.

Motivated from Theorem~\ref{m-orthonormal} we introduce the following definition.

\begin{Def} \label{jointly_orthonormal}
We say that a sequence of pure states 
$(\ket{\psi_i})_{i=1}^D \subseteq (\mathbb{C}^2)^\oplus$ is 
\textbf{jointly orthonormal}, if
and only if for every $m \in \N, \{\ket{\psi_{v_1} \circ \cdots \circ \psi_{v_m}}: 
( v_1, \ldots , v_m ) \in \{ 1, \ldots , D\}^m \} $ is an orthonormal set.


\end{Def}
The next remark clarifies Definition~\ref{jointly_orthonormal}.
\begin{Rmk}
Notice that every jointly orthonormal
sequence is orthonormal (which is obtained by taking $m=1$). However, the converse is not always true. Take the states $\frac{1}{\sqrt 2} (\ket{0} + \ket{00})$ and $\frac{1}{\sqrt 2} (\ket{0} - \ket{00}$ as a counter example.
\end{Rmk}

One can naturally construct a jointly orthonormal sequence by taking each $\ket{\psi_i}$ to be the bit strings produced from a classical uniquely decodable, or a classical prefix-free code. Another possible way to construct these sequences would be by taking the basis of a prefix-free Hilbert space as described in \cite{Muller2008}.

Consider a quantum stochastic source which contains an alphabet of $N$ many pure states $(\ket{s_i})_{i=1}^N$ 
that span a 
Hilbert space $\mathcal{H}$ of dimension $d$. Let $l, m \in \mathbb{N}$ where
$l$ denotes the block size, and $m$ denotes the number 
of blocks. Fix a jointly orthonormal sequence $\{\ket{\psi_i}\}_{i=1}^{d^l}$ such that each $\ket{\psi_i}$ is a length state. Consider the tensor product state $\ket{s_{n_1}} \ket{s_{n_2}} \ldots  \ket{s_{n_{ml}}}$ emitted from the quantum stochastic source, for some 
sequence $n_1,n_2, \ldots , n_{ml} \in \{ 1, \ldots , N \}$. Lastly, consider a sequence of quantum code $(U_k)_{k=1}^m$ where each isometry $U_k$ for $k \in \{1, \ldots, m\}$ is constructed by taking an arbitrary orthonormal basis $\{e_i^k\}_{i=1}^{d^l}$ but the same jointly orthonormal sequence $\{\ket{\psi_i}\}_{i=1}^{d^l}$ such 
that $U_k= \sum_{i=1}^{d^l} \ketbra{\psi_i}{e_i^k}$, and is used to encode the $k^\text{th}$ block consisting of a tensor product of $l$ many pure states. Hence, the tensor product state $\ket{s_{n_1}} \ket{s_{n_2}} \ldots  \ket{s_{n_{ml}}}$ is encoded as the concatenation 
$$
U_1(\ket{s_{n_1}} \cdots \ket{s_{n_l}}) \circ
U_2 (\ket{s_{n_{l+1}} } \cdots \ket{s_{n_{2l}}}) \circ \cdots \circ 
U_m 
(\ket{s_{(m-1)l+1}} \cdots \ket{s_{ml}}).
$$
By the forward direction of quantum Kraft-McMillan Inequality (Theorem~\ref{generalized_Kraft}), the sequence of quantum code $(U_k)_{k=1}^m$ is uniquely decodable, and therefore the concatenation above is well defined.

Since we are interested in minimizing the average codeword length, each isometry will be used to encode the average state of a block as opposed to a specific tensor product of $l$ many pure states. 
Since the probability distribution for each block emitted by a quantum stochastic source is dependent on the entire sequence of states emitted up to the previous block, any $k^\text{th}$-block conditional ensemble state $\rho^{n_1, \ldots, n_{(k-1)l}}$ (as in Definition~\ref{block_ensemble_state}) 
depends on the distinct sequence of states 
$\ket{s_{n_1} \cdots s_{n_{(k-1)l}}}$ which is emitted up to that particular block, (without including it). 
Therefore, we collect the isometries of the form $U^{n_1, \ldots , n_{(k-1)l}} = \sum_{i=1}^{d^l} \ketbra{\psi_i}{e_i^{n_1, \ldots , n_{(k-1)l}}}$, one for each distinct $k^\text{th}$-block conditional ensemble state to form a special block code, as we specify next. 

\begin{Def} \label{special_block_sequence}
Consider a quantum stochastic source $\mathcal{S}$ which contains an alphabet of $N$ many pure states $(\ket{s_i})_{i=1}^N$ 
that span a 
Hilbert space $\mathcal{H}$ of dimension $d$. Let $l, m \in \mathbb{N}$ where
$l$ denotes the block size, and $m$ denotes the number 
of blocks. A \textbf{special block code} is a family of
isometries 
$$
\mathcal{U}= \Big\{ U^{n_1, \ldots , n_{(k-1)l}}: 1 \leq k \leq m, \, n_1, \ldots , n_{(k-1)l} \in \{ 1, \ldots , N \} \Big\},
$$
such that every isometry used in the family $\mathcal{U}$ has a common jointly orthonormal sequence of 
length codewords. Thus more explicitly,  
there exists a jointly orthonormal sequence of  length 
codewords $(\ket{\psi_i})_{i=1}^{d^l} \subseteq (\mathbb{C}^2)^\oplus$, (see Definitions~\ref{jointly_orthonormal} and \ref{def:quantum_code}), and an orthonormal sequence
$(\ket{e_i^{n_1, \ldots , n_{(k-1)l}}})_{i=1}^{d^l}$ for $ 1 \leq k \leq m$ and $n_1, \ldots , n_{(k-1)l} \in \{ 1, \ldots , N \}$ such that 
$$
U^{n_1, \ldots , n_{(k-1)l}} = \sum_{i=1}^{d^l} \ketbra{\psi_i}{e_i^{n_1, \ldots , n_{(k-1)l}}} .
$$
As a special case, for every $k \in \{1, \ldots, m \}$, if $U^{n_1, \ldots, n_{(k-1)l}} = U_k = \sum_{i=1}^{d^l} \ketbra{\psi_i}{e_i^k}$ for all $n_1, \ldots , n_{(k-1)l} \in \{ 1, \ldots , N \}$ where $(\ket{e_i^k})_{i=1}^{d^l}$ is an arbitrary orthonormal sequence fixed for the $k^{th}$ block and independent of the sequence of pure states $\ket{s_{n_1} \cdots s_{n_{(k-1)l}}}$ emitted in the first $k-1$ blocks, then we call such a special block code a \textbf{constrained special block code}. This notion will be used in Section~\ref{Stationary}.
\end{Def}

\begin{Rmk}
Notice that in Definition~\ref{special_block_sequence}, the superscript $n_1, \ldots, n_{(k-1)l}$ which appears in the isometry $U^{n_1, \ldots, n_{(k-1)l}}$ and the orthonormal basis $(\ket{e_i^{n_1, \ldots , n_{(k-1)l}}})_{i=1}^{d^l}$ indicates that these quantities are allowed to depend on the sequence of pure states $\ket{s_{n_1} \cdots s_{n_{(k-1)l}}}$ emitted in the first $k-1$ blocks. In that case, such special block codes are termed as \textbf{adaptive} (or \textbf{dynamic}) special block codes. Our Theorem \ref{Thm:main} states that in order to achieve the optimal lower bound of the average codeword length for encoding $m$ many blocks of states with $l$ many pure states per block emitted by the given quantum stochastic source, one needs to use such an adaptive special block code. In this optimal special block code, the orthonormal basis $(\ket{e_i^{n_1, \ldots , n_{(k-1)l}}})_{i=1}^{d^l}$ for each isometry $U^{n_1, \ldots, n_{(k-1)l}}$ is chosen to coincide with the eigenvectors $(\ket{\lambda_i^{n_1, \ldots, n_{(k-1)l}}})_{i=1}^{d^l}$ of the $k^\text{th}$-block conditional ensemble state $\rho^{n_1, \ldots, n_{(k-1)l}}$.
\end{Rmk}

\begin{Rmk}
The constraint that all the isometries in a given special block code $\mathcal{U}$ have a common sequence of jointly orthonormal length codewords ensures that any sequence 
$(U_k)_{k=1}^m$ of quantum codes, which are chosen from a special block code, is uniquely decodable (by Theorem~\ref{m-orthonormal}). Nevertheless, one could construct more general uniquely decodable quantum codes which do not necessarily meet this requirement, but we do not consider them in this paper. Such an example is given in the remark below.
\end{Rmk}

\begin{Rmk}
Consider a sequence $(U_k)_{k=1}^m$ of $m$ many isometries. Let $O$ and $E$ denote the set of odd and even numbers respectively from $1$ through $m$. 
Let $U^x = \sum_{i=1}^D \ketbra{\psi_i^x}{e_i^x}$ and $U^y = \sum_{i=1}^D \ketbra{\psi_i^y}{e_i^y}$ such that $\ket{\psi_i^x} = \ket{0}^{\otimes i}$ and 
$\ket{\psi_i^y} = \ket{1}^{\otimes i}$ for all $x \in O$, $y \in E$, and $i \in \{1, \ldots, D\}$. It's easy to verify that 
the collection
$$
\Big\{ \ket{\psi_{v_1}^1 \circ \cdots \circ \psi_{v_m}^m} : 
( v_1, \ldots , v_m ) \in \{ 1, \ldots , D\}^m
\Big\}
$$
forms an orthonormal set. Therefore, the sequence $(U_k)_{k=1}^m$ of quantum codes is uniquely decodable by Theorem~\ref{m-orthonormal}. However, the quantum codes $U_k$'s do not have a common sequence of jointly orthonormal length codewords and hence the sequence $(U_k)_{k=1}^m$ is not a special block code.
\end{Rmk}

Now we define the indeterminate length of any state on the 
Fock space.

\begin{Def} \label{Def:indetermine_length}
Define the \textbf{length observable} $\Lambda$ acting on $(\mathbb{C}^2)^\oplus$ to be the unbounded linear operator given by
\begin {equation} \label{length_observ}
\Lambda=\sum_{\ell =0}^\infty \ell \Pi_\ell,
\end{equation}
where $\Pi_\ell$ is the orthogonal projection from
$(\mathbb{C}^2)^\oplus$ onto the subspace 
$(\mathbb{C}^2)^{\otimes \ell}$. 

For any state $\rho$ on the Fock space, define its 
\textbf{indeterminate length} to be the expectation of the length observable with respect to the state $\rho$, i.e.
$$
\Tr (\rho \Lambda),
$$
where $\Tr$ denotes the trace on the Fock space.
\end{Def}

This notion of indeterminate length of a state 
can be used to define
the average codeword length of the special block code.

\begin {Def}\label{average_codeword}
Consider a quantum stochastic source $\mathcal{S}$ 
which contains an alphabet of $N$ many pure states $(\ket{s_i})_{i=1}^N$ 
which span a 
Hilbert space $\mathcal{H}$, and a stochastic process $X$ 
with probability mass function $p$. Let $l, m \in \mathbb{N}$ where
$l$ denotes the block size, and $m$ denotes the number 
of blocks. Consider a special block code
$$
\mathcal{U}=\Big\{ U^{n_1, \ldots , n_{(k-1)l}}: 1 \leq k \leq m, n_1, \ldots , n_{(k-1)l} \in \{ 1, \ldots , N \} \Big\},
$$
as defined in Definition~\ref{special_block_sequence}.
Recall from Definition~\ref{general_ensemble} that the 
ensemble state for the total of $ml$ many emissions by
the quantum stochastic source is equal to
$$
\rho_{ml}=\sum_{n_1,\ldots,n_{ml}=1}^N p(n_1,\ldots,n_{ml})\ketbra{s_{n_1}\cdots s_{n_{ml}}}{s_{n_1}\cdots s_{n_{ml}}} .
$$
We will denote by $\bm{L(\mathcal{U})}$ the \textbf{average codeword length of the special block code} $\bm{\mathcal{U}}$, which is defined 
to be equal to
\begin{align*}
\sum_{n_1,\ldots,n_{ml}=1}^N p(n_1,\ldots,n_{ml})
 & \Tr \Bigg( \ket{U (s_{n_1} \cdots s_{n_{l}}) \circ
U^{n_1, \ldots , n_l} (s_{n_{l+1}} \cdots s_{n_{2l}}) \circ \cdots \circ U^{{n_1,\ldots , n_{(m-1)l}}} (s_{n_{(m-1)l+1}} \cdots s_{n_{ml}})} \\
 &   \bra{
 U(s_{n_1} \cdots s_{n_{l}}) \circ 
 U^{n_1, \ldots , n_{n_l}} 
		(s_{n_{l+1}} \cdots s_{n_{2l}}) \circ \cdots 
\circ U^{{n_1,\ldots , n_{(m-1)l}}} (s_{n_{(m-1)l+1}} \cdots s_{n_{ml}}) } \Lambda \Bigg)
\end{align*}
Let $\bm{ILS(\mathcal{S}, m,l)}$
denote the \textbf{infimum} of the set 
containing  $ L(\mathcal{U})$ for every
\textbf{special} block code $\mathcal{U}$ that is used to encode $ml$ many states emitted by $\mathcal{S}$ into $m$ blocks each of size $l$.
Similarly we define $\bm{ILC(\mathcal{S}, m,l)}$ for \textbf{constrained} special block codes.
\end {Def}

Finally, we are ready to state the main result of our 
article. 
\begin{Thm} \label{Thm:main}
Consider a  quantum stochastic source $\mathcal{S}$ 
consisting of an alphabet of $N$ many pure states
spanning a $d$-dimensional Hilbert space $\mathcal{H}$, 
and a stochastic process $X$ having mass function $p$, 
as in  
Definition~\ref{quantum_source}. Fix $m,l \in \N$. 
Then $ILS(\mathcal{S},m,l)$ can be computed as follows: 
For each $k =1, \ldots , m$,
and a sequence $n_1, \ldots , n_{(k-1)l}$ of integers
chosen from the set $\{ 1, \ldots , N \}$, let 
$(\lambda_i^{n_1, \ldots , n_{(k-1)l}})_{i=1}^{d^l}$ be the 
eigenvalues of the $k^{\text{th}}$ block 
conditional ensemble state $\rho^{n_1, \ldots, n_{(k-1)l}}$,
arranged in decreasing order, and 
$(\ket{\lambda_i^{n_1, \ldots, n_{(k-1)l}}})_{i=1}^{d^l}$ be the corresponding eigenvectors. 

Let 
$$
\mathfrak{L}= \left\{ (\ell_1, \ldots , \ell_{d^l}) : \ell_i \in \N \cup \{ 0 \} \text{ for all }i,\, \ell_1 \leq \ell_2 \leq \cdots \leq \ell_{d^l}, \text{ and } 
\sum_{i=1}^{d^l} 2^{-\ell_i} \leq 1 \right\}.
$$	

Define a function $\mathcal{F_S}: \mathfrak{L} \to [0,\infty)$ by 
\begin{equation} \label{L}
\mathcal{F_S}((\ell_i)_{i=1}^{d^l}) :=
\sum_{j=2}^{m}\left(\sum_{n_1,\ldots ,n_{(j-1)l}=1}^N
p(n_1,\ldots ,n_{(j-1)l}) \sum_{i= 1}^{d^l} \lambda_{i}^{n_1,\ldots ,n_{(j-1)l}} \ell_i \right) + \sum_{i=1}^{d^l} \lambda_i \ell_i.
\end{equation}

Then,  
\begin{equation} \label{E:Best_ACL}
ILS(\mathcal{S},m,l)= \min\{ \mathcal{F_S}((\ell_i)_{i=1}^{d^l}) : (\ell_i)_{i=1}^{d^l} \in \mathfrak{L} \}.
\end{equation}

Moreover, the infimum defining $ILS(\mathcal{S},m,l)$ is actually a minimum,
i.e., there exists a special block code
$$
\mathcal{V}= \Big\{ V^{n_1, \ldots , n_{(k-1)l}}: 
k \in \{ 1, \ldots , m \},\text{ and }
n_1, \ldots ,n_{(k-1)l} \in \{ 1, \ldots , N \} \Big\},
$$ 
which can be used to encode $ml$ many states emitted by 
$\mathcal{S}$ into $m$ blocks each of size $l$ such that 
\begin{equation} \label{E:optimality}
\min\{ \mathcal{F_S}((\ell_i)_{i=1}^{d^l}) : (\ell_i)_{i=1}^{d^l} \in \mathfrak{L} \} = L(\mathcal{V}).
\end{equation}
The minimizer $\mathcal{V}$ is given as follows: 
Assume that $\mathcal{F_S}$ achieves its minimum on $\mathfrak{L}$ 
at the point $(\ell_i)_{i=1}^{d^l} \in \mathfrak{L}$. 
Since the sequence $(\ell_i)_{i=1}^{d^l}$ satisfies the 
classical Kraft-McMillan inequality, (which is the last condition in the definition of $\mathfrak{L}$), by the converse of classical Kraft-McMillan inequality
there exists a sequence 
$(\omega_i)_{i=1}^{d^l}$ of classical bit strings with 
corresponding lengths $(\ell_i)_{i=1}^{d^l}$ such that the sequence $(\omega_i)_{i=1}^{d^l}$ forms the image of a classical uniquely decodable code. Then the 
corresponding sequence of qubit strings $(\ket{\omega_i})_{i=1}^{d^l}$ 
forms an optimal jointly orthonormal sequence of length codewords. 
For each $k \in \{ 1, \ldots , m \}$, and string 
$n_1, \ldots ,n_{(k-1)l} \in \{ 1, \ldots , N \}$, 
define
$$
V^{n_1, \ldots , n_{(k-1)l}} : \mathcal{H}^{\otimes l} \to (\mathbb{C}^2)^\oplus,
$$
by
$$
V^{n_1, \ldots , n_{(k-1)l}} = \sum_{i=1}^{d^l} \ketbra{\omega_i}{\lambda_i^{n_1, \ldots, n_{(k-1)l}}}.
$$
\end{Thm}
\begin{proof}[Sketch of the proof] 
Here's the outline of the proof of the above theorem.
Fix $m$ and $l$ as the number of blocks and the block size respectively.  Let 
$$
\mathfrak{L}= \left\{ (\ell_1, \ldots , \ell_{d^l}) : \ell_i \in \N \cup \{ 0 \} \text{ for all }i,\, \ell_1 \leq \ell_2 \leq \cdots \leq \ell_{d^l}, \text{ and } 
\sum_{i=1}^{d^l} 2^{-\ell_i} \leq 1 \right\}.
$$
 Fix a special block code $\mathcal{U}$ (Definition~\ref{special_block_sequence}) such that all isometries in $\mathcal{U}$ have $(\ket{\psi_i})_{i=1}^{d^l}$ as their common jointly orthonormal sequence of length codewords. Suppose that length$(\ket{\psi_i}) = \ell_i \in \N $ and without loss of generality, assume $\ell_r \leq \ell_s$ for $r \leq s$. By the forward direction of quantum Kraft-McMillan inequality (Theorem~\ref{generalized_Kraft}), we have that the set of non-negative integer lengths $\{\ell_i\}_{i=1}^{d^l}$ belongs to $\mathfrak{L}$.

For a specific sequence of pure states $\ket{s_{n_1}, \cdots s_{n_{(k-1)l}} }$ as the first $(k-1)l$ emissions, consider the corresponding $k^{th}$ block conditional ensemble state $\rho^{n_1, \ldots, n_{(k-1)l}}$. Also, consider the spectral decomposition of $\rho^{n_1, \ldots, n_{(k-1)l}}$ as
$\rho^{n_1, \ldots, n_{(k-1)l}} = \sum_{i=1}^{d^l} \lambda_i^{n_1, \ldots, n_{(k-1)l}} \ketbra{\lambda_i^{n_1, \ldots, n_{(k-1)l}}}$ and without loss of generality, assume that $\lambda_r^{n_1, \ldots, n_{(k-1)}} \geq \lambda_s^{n_1, \ldots, n_{(k-1)l}}$ for $r \leq s$. The codeword length of encoding the state $\rho^{n_1, \ldots, n_{(k-1)l}}$ by the quantum code $U^{n_1, \ldots, n_{(k-1)l}} = \sum_{i=1}^{d^l} \ketbra{\psi_i}{e_i^{n_1, \ldots, n_{(k-1)l}}} \in \mathcal{U} $, which has the length codewords $(\ket{\psi_i})_{i=1}^{d^l}$ and an arbitrary orthonormal basis $(\ket{e_i^{n_1, \ldots, n_{(k-1)l}}})_{i=1}^{d^l}$,  is given by 
\begin{equation} \label{conditional_ensemble_length}
\Tr(U^{n_1, \ldots, n_{(k-1)l}} \rho^{n_1, \ldots, n_{(k-1)l}} (U^{n_1, \ldots, n_{(k-1)l}})^\dagger \Lambda).
\end{equation} 
Then, we show by using the Birkhoff-von Neumann theorem and Lemma \ref{lemma} that the codeword length given by the expression~\eqref{conditional_ensemble_length} is minimized for the given set of lengths $\{\ell_i\}_{i=1}^{d^l}$ when $U^{n_1, \ldots, n_{(k-1)l}} = \sum_{i=1}^{d^l} \ketbra{\psi_i}{\lambda_i^{n_1, \ldots, n_{(k-1)l}}}$. In this case, expression~\eqref{conditional_ensemble_length} can be re-written as
\begin{equation}
\Tr(U^{n_1, \ldots, n_{(k-1)l}} \rho^{n_1, \ldots, n_{(k-1)l}} (U^{n_1, \ldots, n_{(k-1)l}})^\dagger \Lambda) = \sum_{i=1}^{d^l} \lambda_i^{n_1, \ldots, n_{(k-1)l}} \ell_i.
\end{equation}
Since the state $\rho^{n_1, \ldots, n_{(k-1)l}}$ occurs with the probability $p( n_1, \ldots, n_{(k-1)l})$, the minimum average codeword length for encoding the $k^{th}$ block ensemble state $\rho^k = \sum_{n_1, \ldots, n_{(k-1)l} =1}^N p(n_1, \ldots, n_{(k-1)l}) \rho^{n_1, \ldots, n_{(k-1)l}}$ (Definition \ref{block_ensemble_state}) is given by $\sum_{n_1, \ldots, n_{(k-1)l}} p(n_1, \ldots, n_{(k-1)l}) \sum_{i=1}^{d^l} \lambda_i^{n_1, \ldots, n_{(k-1)l}} \ell_i$ for the fixed set of lengths $\{\ell_i\}_{i=1}^{d^l}$. Recall that for $k=1$ we use the convention that $\lambda_i^{n_1, \ldots, n_{(k-1)l}} = \lambda_i$ and $p(n_1, \ldots, n_{(k-1)l}) = 1$. Summing the minimum average codeword lengths for each block gives the expression in Equation~(\ref{L}), which is the minimum average codeword length for $m$ blocks. It is worth noticing that the minimum average codeword length for $m$ blocks is additive of the minimum average codeword lengths for each block. Refer to the paragraph following Equation~\eqref{E:begining_of_m-1_block_part_2} in the Appendix \ref{main_theorem} for the details. Finally, minimization of Equation~(\ref{L}) over all sets of  lengths  $(\ell_i)_{i=1}^{d^l} \in \mathfrak{L}$
 gives Equation~(\ref{E:Best_ACL}) which is the minimum average codeword length over the set of all special block codes for encoding $m$ blocks. Once the minimizer set of lengths is found, one can then construct an optimal jointly orthonormal sequence of length codewords $(\ket{\omega_i})_{i=1}^{d^l}$ and define a minimizer special block code $\mathcal{V}$ as stated in the theorem.
See Appendix \ref{main_theorem} for the full proof of the theorem. 
\end{proof}

Our Theorem~\ref{Thm:main} recovers the result of 
\cite[Theorem 2]{Bellomo2017Nov} which was the case for $m=1$, (i.e.\ when the whole message is encoded as a single block).

\section{A simplified compression for stochastic quantum sources and its application to stationary quantum sources} \label{Stationary}
Notice from the statement of Theorem~\ref{Thm:main} that achieving the above theoretical lower bound requires that for every $k^{th}$ block, Alice (the encoder) smartly chooses the optimal isometry $V^{n_1, \ldots, n_{(k-1)l}} = \sum_{i=1}^{d^l} \ketbra{\omega_i}{\lambda_i^{n_1, \ldots, n_{(k-1)l}}}$ to encode each $k^{th}$ block conditional ensemble state $\rho^{n_1, \ldots,n_{(k-1)l}}$, which requires an accurate estimation of the exact sequence of pure states $\ket{s_{n_1, \ldots, n_{(k-1)l}}}$ emitted in the first $(k-1)$ blocks. Since the pure states emitted by the source are not necessarily orthogonal, this task of state discrimination of non-orthogonal states can be difficult. So, we relax this requirement and envision a simplified compression technique where for a given $k^{th}$ block, Alice uses a fixed isometry $U_k$ to encode every $k^{th}$ block conditional ensemble state $\rho^{n_1, \ldots, n_{(k-1)l}}$ regardless of the emissions in the first $(k-1)$ blocks. In other words, we would like to restrict our attention to the set of  constrained special block codes (Definition~\ref{special_block_sequence}) where $U^{n_1, \ldots, n_{(k-1)l}} = U_k= \sum_{i=1}^{d^l} \ketbra{\psi_i}{e_i^k}$ for all $n_1, \ldots, n_{(k-1)l} \in \{1, \ldots, N\}$. We then eventually show that when such a compression is applied to the pure states emitted from a stationary quantum source and one appropriately chooses the isometries $U_k$'s, the optimal lower bound of the average codeword length per symbol equals the von Neumann entropy rate of the source for an asymptotically long block size. From here on, in order to distinguish between the two compressions, we will refer to the compression outlined in Theorem~\ref{Thm:main} as the \textbf{adaptive} compression while this new simplified compression that we discuss below as the \textbf{non-adaptive} compression. It is worth emphasizing that the only difference between the two compressions is that the adaptive compression minimizes the average codeword length over the set of all special block codes, while the non-adaptive compression minimizes the same over the set of all constrained special block codes. 

\begin{Thm} \label{Thm:derived}
Consider a  quantum stochastic source $\mathcal{S}$ 
consisting of an alphabet of $N$ many pure states
spanning a $d$-dimensional Hilbert space $\mathcal{H}$, 
and a stochastic process $X$ having mass function $p$, 
as in  
Definition~\ref{quantum_source}. Fix $m,l \in \N$. 
Then $ILC(\mathcal{S},m,l)$ can be computed as follows: 
For each $k =1, \ldots , m$, let 
$(\lambda_i^k)_{i=1}^{d^l}$ be the 
eigenvalues of the $k^{\text{th}}$ block 
ensemble state $\rho^k$,
arranged in decreasing order, and 
$(\ket{\lambda_i^k})_{i=1}^{d^l}$ be the corresponding eigenvectors. 

Let 
$$
\mathfrak{L}= \left\{ (\ell_1, \ldots , \ell_{d^l}) : \ell_i \in \N \cup \{ 0 \} \text{ for all }i,\, \ell_1 \leq \ell_2 \leq \cdots \leq \ell_{d^l}, \text{ and } 
\sum_{i=1}^{d^l} 2^{-\ell_i} \leq 1 \right\}.
$$	

Define a function $\mathcal{F_C}: \mathfrak{L} \to [0,\infty)$ by 
\begin{equation} \label{LC}
\mathcal{F_C}((\ell_i)_{i=1}^{d^l}) :=
\sum_{k=1}^m \sum_{i=1}^{d^l} \lambda_i^k \ell_i.
\end{equation}

Then,  
\begin{equation} \label{E:Best_ACL_C}
ILC(\mathcal{S},m,l)= \min\{ \mathcal{F_C}((\ell_i)_{i=1}^{d^l}) : (\ell_i)_{i=1}^{d^l} \in \mathfrak{L} \}.
\end{equation}

Moreover, the infimum defining $ILC(\mathcal{S},m,l)$ is actually a minimum,
i.e., there exists a constrained special block code
$$
\mathcal{V}= \Big\{ V^k: 
k \in \{ 1, \ldots , m \} \Big\},
$$ 
which can be used to encode $ml$ many states emitted by 
$\mathcal{S}$ into $m$ blocks each of size $l$ such that 
\begin{equation} \label{E:optimality_C}
\min\{ \mathcal{F_C}((\ell_i)_{i=1}^{d^l}) : (\ell_i)_{i=1}^{d^l} \in \mathfrak{L} \} = L(\mathcal{V}).
\end{equation}
The minimizer $\mathcal{V}$ is given as follows: 
Assume that $\mathcal{F_C}$ achieves its minimum on $\mathfrak{L}$ 
at the point $(\ell_i)_{i=1}^{d^l} \in \mathfrak{L}$. 
Since the sequence $(\ell_i)_{i=1}^{d^l}$ satisfies the 
classical Kraft-McMillan inequality, (which is the last condition in the definition of $\mathfrak{L}$), by the converse of classical Kraft-McMillan inequality
there exists a sequence 
$(\omega_i)_{i=1}^{d^l}$ of classical bit strings with 
corresponding lengths $(\ell_i)_{i=1}^{d^l}$ such that the sequence $(\omega_i)_{i=1}^{d^l}$ forms the image of a classical unqiuely decodable code. Then the 
corresponding sequence of qubit strings $(\ket{\omega_i})_{i=1}^{d^l}$ 
forms an optimal jointly orthonormal sequence of length codewords. 
For each $k \in \{ 1, \ldots , m \}$, 
define
$$
V^k : \mathcal{H}^{\otimes l} \to (\mathbb{C}^2)^\oplus,
$$
by
$$
V^k = \sum_{i=1}^{d^l} \ketbra{\omega_i}{\lambda_i^k}.
$$
\end{Thm}

The following is a corollary of Theorem~\ref{Thm:derived}.
\begin{Cor}\label{corollary_stationary}
Consider a stationary quantum stochastic source $\mathcal{S}$ (Definition~\ref{stationary_source}). We are interested in deriving the minimum average codeword length for $m$ blocks each of block size $l$ using constrained special block codes, i.e. $ILC(\mathcal{S},m,l)$ as in Theorem~\ref{Thm:derived}. Let the spectral decomposition of the ensemble state $\rho_l$ (or equivalently, $\rho^1$) (Definition~\ref{general_ensemble}) for $l$ many emissions from the source to be:
$$
\rho_{l} = \sum_{i=1}^{d^l} \lambda_i \ketbra{\lambda_i}.
$$
Consider the isometry $V = \sum_{i=1}^{d^l} \ketbra{\omega_i}{\lambda_i}$, where $(\omega_i)_{i=1}^{d^l}$ is the sequence of bit strings obtained from the Huffman code based on the probability distribution $(\lambda_i)_{i=1}^{d^l}$ such that  length$(\omega_i) = \ell_i$. 
Then the minimizer constrained special block code $\mathcal{V}$ for encoding $m$ blocks comprises of the identical quantum code given by $\mathcal{V} = V$ and the minimum average codeword length per symbol equals the von-Neumann entropy rate of the source for an asymptotically long block size, i.e. $\lim_{l \to \infty} \frac{ILC(\mathcal{S},1,l)}{l} = \lim_{l \to \infty} \frac{\mathcal{S}(\rho_{l})}{l}$.
\end{Cor}
Refer to Appendix~\ref{additional_theorems} for the proofs of Theorem~\ref{Thm:derived} and Corollary~\ref{corollary_stationary}. \\
\section{Conclusion and open questions}
Bellomo, Bosyk, Holik and Zozor \cite{Bellomo2017Nov} defined the optimal quantum code for
encoding the entire sequence of 
pure states emitted by a given quantum stochastic 
source 
into a single block 
and showed that the associated minimum average codeword length could be bounded from above and below in terms of the von-Neumann
entropy of the source. The present work focuses on the encoding of the sequence of $q$ many pure states from the quantum stochastic 
source $\mathcal{S}$ into $m \in \N$ blocks, each of block size $l \in \N$ such that $q=ml$ using special block codes, and we gave an optimal lower bound (denoted by $ILS(\mathcal{S},m,l)$ in Theorem~\ref{Thm:main}) of the average codeword length computed over all special block codes. Lastly, we gave a simplified non-adaptive compression technique based on constrained special block codes for quantum stochastic sources that emit pure states where the encoder can fix one isometry per block of states, and showed that if one uses such a compression technique for a stationary quantum source, then the optimal lower bound of the average codeword length per symbol over all constrained special block codes (denoted by $ILC(\mathcal{S},m,l)$ in Corollary~\ref{corollary_stationary}) equals the von-Neumann entropy rate of the source for an asymptotically long block size. Since we encode the pure states into indeterminate-length codewords which may not have a concrete physical interpretation, whether the derived optimal lower bounds are practically achievable require further investigation. Nevertheless, these theoretical lower bounds can serve as a benchmark for measuring the efficiency of lossless quantum block compression algorithms.  
The following interesting questions are still open.

For a fixed length message consisting of $q \in \N$ many pure states emitted by a given quantum stochastic source S, what are the optimal values of $m$ and $l$ such that $ml = q$ and $ILS(\mathcal{S},m,l)$ over special (or $ILC(\mathcal{S},m,l)$ over constrained special) block codes is minimized? In other words, what is the optimal number of blocks and the optimal block size for encoding a fixed length message of pure states emitted by a quantum stochastic source using special (or constrained special) block codes?

It would be interesting to relate
$ILS(\mathcal{S},m,l)$  
over special block codes with an entropic quantity of a state 
related to the quantum source. 
Perhaps one needs an alternate entropic quantity (other than von~Neumann entropy) in order to tightly bound $ILS(\mathcal{S},m,l)$, (see \cite{Ahlswede2004Jun, Anshu2016Sep}).

Notice that special block codes are not the most general uniquely decodable quantum codes.
So, it would be interesting to 
derive a similar optimal lower bound of the average codeword length over all uniquely decodable quantum codes used for block encoding the sequence of pure states emitted by a quantum stochastic source.

One could also pursue answering these questions for a quantum stochastic source which emits mixed states instead of pure states, which is yet another interesting direction of research.

\bibliographystyle{abbrvurl}
\bibliography{Bibliography}

\appendix
\section {A characterization of uniquely decodable quantum codes} \label{proof_uniquely_decodable}
\begin{proof}[Proof of Theorem~\ref{m-orthonormal}]
The \lq\lq only if\rq\rq\ direction:\\
Suppose the sequence of quantum codes  $(U_k)_{k=1}^m$ is uniquely decodable, and $U_k = \sum_{i=1}^D \ketbra{\psi_i^k}{e_i^k}$. Then by Definition~\ref{uniquely_decodable}, the linear map $U_1 \circ \cdots \circ U_m: \mathcal{K}^{\otimes m} \to (\mathbb{C}^2)^\oplus$ is an isometry. Therefore, it follows that
\begin{align} \nonumber
&\phantom{....}\braket{(U_1 \circ \cdots \circ U_m)(e_{v_1}^1 \otimes  \cdots  \otimes e_{v_m}^m)}{(U_1 \circ \cdots \circ U_m)(e_{{v_1}^\prime}^1 \otimes \cdots \otimes e_{{v_m}^\prime}^m)} \\ \nonumber
&= \braket{e_{v_1}^1 \otimes  \cdots  \otimes e_{v_m}^m}{e_{{v_1}^\prime}^1 \otimes \cdots \otimes e_{{v_m}^\prime}^m} \\ \label{orthonormal_eqn}
&= \delta_{v_1, {v_1}^\prime} \cdots \delta_{v_m, {v_m}^\prime}.
\end{align}

On the other hand, by Definition~\ref{concatenation}, $(U_1 \circ \cdots \circ U_m)\ket{e_{v_1}^1 \otimes  \cdots  \otimes e_{v_m}^m}=\psi_{v_1}^1 \circ  \cdots  \circ \psi_{v_m}^m$. 
Therefore, by Equation~\eqref{orthonormal_eqn} we obtain that
\begin{align*}
&\phantom{....}\braket{\psi_{v_1}^1 \circ \cdots  \circ \psi_{v_m}^m}{\psi_{{v_1}^\prime}^1 \circ \cdots \circ \psi_{{v_m}^\prime}^m} \\
&=\braket{(U_1 \circ \cdots \circ U_m)(e_{v_1}^1 \otimes  \cdots  \otimes e_{v_m}^m)}{(U_1 \circ \cdots \circ U_m)(e_{{v_1}^\prime}^1 \otimes \cdots \otimes e_{{v_m}^\prime}^m)} \\
&= \delta_{v_1, {v_1}^\prime} \cdots \delta_{v_m, {v_m}^\prime},
\end{align*}
i.e.\ the collection
$
\Big\{ \ket{\psi_{v_1}^1 \circ \cdots \circ \psi_{v_m}^m} : 
( v_1, \ldots , v_m ) \in \{ 1, \ldots , D\}^m
\Big\}
$
is an orthonormal set. \\
The \lq\lq if\rq\rq\ direction:\\
Suppose that the collection $
\Big\{ \ket{\psi_{v_1}^1 \circ \cdots \circ \psi_{v_m}^m} : 
( v_1, \ldots , v_m ) \in \{ 1, \ldots , D\}^m
\Big\}
$
is an orthonormal set.
Consider any two sequences of vectors 
$(\ket{\phi_k})_{k=1}^m$, and  $(\ket{\Phi_k})_{k=1}^m$ 
such that $\ket{\phi_k}$ and $\ket{\Phi_k}$ belong in the 
range of $U_k$ for $k=1,\ldots ,m$. Hence, both 
$\ket{\phi_k}$ and $\ket{\Phi_k}$ belong  in the linear span 
of $\{\ket{\psi_v^k}\}_{v=1}^D$ for every 
$k \in {1, \ldots, m}$. Thus, there exist scalars  
$(a_{v_k}^k)_{v_k=1}^D$
 and $(A_{v_k}^k)_{v_k=1}^D$ such that
$$
\ket{\phi_k} = \sum_{v_k=1}^D a_{v_k}^k \ket{\psi_{v_k}^k},
\quad \text{and} \quad  
\ket{\Phi_k} = \sum_{v_k=1}^D A_{v_k}^k \ket{\psi_{v_k}^k}. 
$$
Therefore,
$$
\ket{\phi_1 \circ \cdots \circ \phi_m} = 
\left( \sum_{v_1=1}^D a_{v_1}^1 \ket{\psi_{v_1}^1} \right) \circ \cdots \circ \left( \sum_{v_m=1}^D a_{v_m}^m \ket{\psi_{v_m}^m} \right) = \sum_{v_1, \ldots , v_m=1}^D a_{v_1}^1 \cdots a_{v_m}^m \ket{\psi_{v_1}^1 \circ \cdots \circ \psi_{v_m}^m}.
$$
and similarly,
$$
\ket{\Phi_1 \circ \cdots \circ \Phi_m} =
\sum_{v_1', \ldots , v_m'=1}^D A_{v_1'}^1 \cdots A_{v_m'}^m \ket{\psi_{v_1'}^1 \circ \cdots \circ \psi_{v_m'}^m}.
$$
Thus,
$$
\braket{\phi_1 \circ \cdots \circ \phi_m}{\Phi_1 \circ \cdots \circ \Phi_m}=
\sum_{\substack{v_1, \ldots , v_m=1\\v_1^\prime, \ldots , v_m^\prime=1}}^D 
\overline{a_{v_1}^1} \cdots \overline{a_{v_m}^m} 
A_{v_1'}^1 \cdots A_{v_m'}^m 
\braket{\psi_{v_1}^1 \circ \cdots \circ \psi_{v_m}^m}{\psi_{v_1^\prime}^1 \circ \cdots \circ \psi_{v_m^\prime}^m}.
$$
Since the states $\ket{\psi_{v_1}^1 \circ \cdots \circ \psi_{v_m}^m}$ 
form an orthonormal set,  we have

\begin{align} \nonumber
\braket{\phi_1 \circ \cdots \circ \phi_m}{\Phi_1 \circ \cdots \circ \Phi_m}
&=
\sum_{\substack{v_1, \ldots , v_m=1\\v_1^\prime, \ldots , v_m^\prime=1}}^D 
\overline{a_{v_1}^1} \cdots \overline{a_{v_m}^m} 
A_{v_1^\prime}^1 \cdots A_{v_m^\prime}^m 
\delta_{(v_1, \ldots, v_m),(v_1^\prime, \ldots , v_m^\prime)}\\ \nonumber
&= 
\sum_{v_1,v_1^\prime=1}^D \overline{a_{v_1}^1} A_{v_1^\prime}^1 \delta_{v_1,v_1^\prime} \cdots \sum_{v_m,v_m^\prime=1}^D \overline{a_{v_m}^1} A_{v_m^\prime}^1 \delta_{v_m,v_m^\prime} \\ \nonumber
&= 
\sum_{v_1,v_1^\prime=1}^D \overline{a_{v_1}^1} A_{v_1^\prime}^1 \braket{\psi_{v_1}}{\psi_{v_1^\prime}} \cdots \sum_{v_m,v_m^\prime=1}^D \overline{a_{v_m}^1} A_{v_m^\prime}^1 \braket{\psi_{v_m}}{\psi_{v_m^\prime}} \\  \label{concat_tensor}
&= \braket{\phi_1}{\Phi_1} \cdots \braket{\phi_m}{\Phi_m}.
\end{align}

Since $\ket{\phi_k}$ and $\ket{\Phi_k}$ belong in the 
range of $U_k$ for $k=1,\ldots ,m$, there exist 
 vectors $(\ket{\chi_k})_{k=1}^m$,  $(\ket{X_k})_{k=1}^m$ in $\mathcal{K}$
such that $\phi_k = U_k (\chi_k)$ and $\Phi_k= U_k (X_k)$
for all 
$k=1, \ldots, m$. Hence, 
\begin{align}
\braket{(U_1 \circ \cdots \circ U_m)(\chi_1 \otimes \cdots \otimes \chi_m)}{(U_1 \circ \cdots \circ U_m)(X_1 \otimes \cdots \otimes X_m)} &= 
\braket{U_1 \chi_1 \circ \cdots \circ U_m \chi_m }{U_1 X_1 \circ \cdots \circ U_m X_m} \nonumber \\
&=\braket{U_1\chi_1}{U_1X_1} \cdots \braket{U_m \chi_m}{U_mX_m} \nonumber \\
&= \braket{\chi_1}{X_1} \cdots \braket{\chi_m}{X_m}, \label{E:inner_product_breaks}
\end{align}
(where the first equality is valid because of the definition
of $U_1 \circ \cdots \circ U_m$, the second equality is valid because of Equation~\eqref{concat_tensor}, and the third equality is 
valid since $U_1, \ldots , U_m$ are isometries).
Thus, $U_1 \circ \cdots \circ U_m$ maps an orthonormal basis 
of $\mathcal{K}^{\otimes m}$ to an orthonormal set of 
$(\mathbb{C}^2)^\oplus$, and thus it is an isometry. In particular, by taking $\phi_k=\Phi_k$, (and hence $\chi_k=X_k$) for $k=1, \ldots , m$, Equation~\eqref{E:inner_product_breaks} gives that 
$U_1 \circ \cdots \circ U_m$ satisfies
$$
\norm{(U_1 \circ \cdots \circ U_m )(\chi_1 \otimes \cdots \otimes \chi_m)}=
\norm{\ket{\chi_1}} \cdots \norm{\ket{\chi_m}}=
\norm{\chi_1 \otimes \cdots \otimes \chi_m}.
$$
\end{proof}

Notice that Equation~\eqref{concat_tensor} gives the following important consequence:

\begin{Rmk} \label{rmk:well-definition_of_concatenation}
Let $(\ket{\psi_i})_{i=1}^D$ be a jointly orthonormal sequence in the Fock space, and $m$ be a fixed natural number. Then, on the linear span of $(\ket{\psi_i})_{i=1}^D$ concatenations of $m$-many vectors are well defined.	
\end{Rmk}

\section{The proof of Theorem 2.17} \label{main_theorem}

First, we need to introduce a well known lemma which is required in the proof of our main result.

\begin{Lem} \label{lemma}
For $z \in \N$ consider a non-increasing sequence of positive 
real numbers $Q_1 \geq Q_2 \geq \cdots \geq Q_z \geq 0$. 
Further, consider another arbitrary sequence of positive real numbers $l_1, l_2, \ldots, l_z$ and its 
non-decreasing enumeration 
$l_1^\prime \leq l_2^\prime \leq \cdots \leq l_z^\prime$. Then,
$$
\sum_{i=1}^z Q_i l_i^\prime \leq \sum_{i=1}^z Q_i l_i.
$$
\end{Lem}

\begin{proof}
Consider the two terms $(Q_1L + ql_1^\prime)$, (associated 
to $Q_1$ and to $l_1^\prime$), in the sum $\sum_{i=1}^z Q_i l_i$. Since $Q_1$ is the largest 
term in the sequence $\{Q_i\}_{i=1}^z$, we have $Q_1 \geq q$. Also, since $l_1^\prime$ is the smallest term in the 
sequence $\{l_i\}_{i=1}^z$, we have $l_1^\prime \leq L$. 
Thus, 
$$
 Q_1(L-l_1^\prime) \geq q(L-l_1^\prime) 
\Rightarrow  Q_1L + ql_1^\prime \geq Q_1l_1^\prime + qL .
$$
Therefore, by replacing the terms $(Q_1L + ql_1^\prime)$ in 
the sum $\sum_{i=1}^z Q_i l_i$ with the smaller quantity 
$(Q_1l_1^\prime + qL)$ yields a smaller sum 
$\sum_{i=1}^z Q_i l_i - (Q_1L + ql_1^\prime) + (Q_1l_1^\prime + qL)$. By continuing similarly for the next largest number $Q_2$ in the newly created sum $\sum_{i=1}^z Q_i l_i - (Q_1L + ql_1^\prime) + (Q_1l_1^\prime + qL)$  yields an even smaller sum. 
By continuing this process for all $Q_i$'s one by one in descending order, we obtain the sum $\sum_{i=1}^z Q_i l_i^\prime$ which is smaller than the original sum $\sum_{i=1}^z Q_i l_i$.
\end{proof}

Another preliminary material that we need for our main 
result is the following theorem which is a straightforward 
generalization of the quantum Kraft-McMillan inequality (and 
its converse) considered in \cite{Androulakis2019Sep}.

\begin{Thm} \label{generalized_Kraft}
Consider a Hilbert space $\mathcal{H}$ of dimension $D$ and a sequence of jointly orthonormal length codewords $(\ket{\psi_i})_{i=1}^D \subseteq (\C^2)^{\oplus}$ such that each $\ket{\psi_i} \in (\mathbb{C}^2)^{\otimes \ell_i}$ for some $\ell_i \in \N$. For $k \in \{1,\ldots, m\}$, if each $U_k = \sum_{i=1}^D \ketbra{\psi_i}{e_i^k} $, then the finite sequence of quantum codes
$(U_k)_{k=1}^m$
is uniquely decodable and the following inequality is satisfied
\begin{equation} \label{kraft_inequality}
\sum_{i=1}^D 2^{-\ell_i} \leq 1.
\end{equation}
Conversely, assume that 
$(U_k)_{k=1}^m$ is a sequence of quantum codes such that each $U_k = \sum_{i=1}^D \ketbra{\psi_i}{e_i^k}$, where $(\ket{\psi_i})_{i=1}^D \subseteq (\C^2)^{\oplus}$ is the
common sequence of length codewords such that each $\ket{\psi_i} \in (\mathbb{C}^2)^{\otimes \ell_i}$ for some $\ell_i \in \N$ and the lengths $(\ell_i)_{i=1}^D$ satisfy the Inequality~\eqref{kraft_inequality}.
Then, there exist a jointly orthonormal sequence of 
length codewords $(\ket{\psi_i^\prime})_{i=1}^D \subseteq (\C^2)^{\oplus}$
such that each $\ket{\psi_i'} \in (\mathbb{C}^2)^{\otimes \ell_i}$, and a uniquely decodable quantum code $(U_k^\prime)_{k=1}^m$ such that each $U_k^\prime = \sum_{i=1}^D \ketbra{\psi_i^\prime}{e_i^k}$.
\end{Thm}
\begin{proof}
Fix a jointly orthonormal sequence $\{\ket{\psi_i}\}_{i=1}^{D}$ such that each $\ket{\psi_i}$ is a length state i.e. $\Tr(\ketbra{\psi_i} \Lambda) = \ell_i \in \N$ or equivalently, $\ket{\psi_i} \in (\mathbb{C}^2)^{\otimes \ell_i}$. Consider a sequence of quantum codes $(U_k)_{k=1}^D$ where each isometry $U_k$ for $k \in \{1, \ldots, m\}$ is constructed by taking an arbitrary orthonormal basis $\{e_i^k\}_{i=1}^{D}$ but the same jointly orthonormal sequence $\{\ket{\psi_i}\}_{i=1}^{D}$ such 
that $U_k= \sum_{i=1}^D \ketbra{\psi_i}{e_i^k}$. By Definition~\ref{jointly_orthonormal}, it follows that the collection $\{\ket{\psi_{i_1} \circ \psi_{i_2} \circ \cdots \circ  \psi_{i_m}}\}$ for $\{i_1, \ldots, i_m\} \in \{1,\ldots,D\}$ is an orthonormal set. Then, by Theorem~\ref{m-orthonormal}, it follows that the sequence $(U_k)_{k=1}^D$ is a uniquely decodable quantum code. \\
For each $n, N \in \N$, let
$$
C_n^N = \{\ket{\Psi^n} \in (\mathbb{C}^2)^{\otimes N} : \ket{\Psi^n} = \ket{\psi_{i_1} \circ \psi_{i_2} \circ \cdots \circ  \psi_{i_n}} \ \text{for some} \  i_1, \ldots, i_n \in \{1,\ldots,D\} \}
$$
be the collection of states formed from concatenation of $n$-many codewords and having length $N$, and let
$$
d_\ell=\# \{i \in \{1, \ldots, D\} : \psi_i \in (\mathbb{C}^2)^{\otimes \ell}\}
$$ 
be the number of $\ket{\psi_i}$'s in 
$(\mathbb{C}^2)^{\otimes \ell}$ for each $\ell \in \N$. 
Since the sequence $(\ket{\psi_i})_{i=1}^D$ is jointly orthonormal, the elements of 
$C_n^N$ are 
pairwise orthogonal, hence linearly independent, and since 
they belong in $(\C^2)^{\otimes N}$ whose dimension is equal 
to $2^N$, we have
 $$
\# C_n^N =\sum_{\ell_{i_1} + \cdots + \ell_{i_n}=N} d_{\ell_{i_1}} d_{\ell_{i_2}} \ldots d_{\ell_{i_n}} \leq 2^N.
$$

Thus,
$$
2^{-N}\sum_{\ell_{i_1} + \cdots + \ell_{i_n}=N} d_{\ell_{i_1}} d_{\ell_{i_2}} \ldots d_{\ell_{i_n}} = \sum_{\ell_{i_1} + \cdots + \ell_{i_n}=N} (2^{-\ell_{i_1}} d_{\ell_{i_1}})(2^{-\ell_{i_2}} d_{\ell_{i_2}}) \cdots (2^{-\ell_{i_n}} d_{\ell_{i_n}}) \leq 1
$$
Set $\ell_{max} = \max\limits_{1 \leq i \leq D} \{\ell_i\}$ so that $N \leq n\ell_{max}$. 
Summing the above inequality over $N$ we obtain
$$
\sum_{\ell_{i_1}, \ell_{i_2}, \ldots, \ell_{i_n} =1} ^ {\ell_{max}} = (2^{\ell_{i_1}} d_{\ell_{i_1}})(2^{\ell_{i_2}} d_{\ell_{i_2}}) \cdots (2^{\ell_{i_n}} d_{\ell_{i_n}}) = \left(\sum_{\ell=1}^{\ell_{max}}2^{-\ell}{d_\ell}\right)^n \leq n\ell_{max}.
$$
Notice that the left hand side of this inequality is exponential while the right hand side is linear. This implies that the left hand side is bounded above by 1. Hence, we must have that 
$$
\sum_{\ell = 1}^{\ell_{max}} 2^{-\ell}d_{\ell} = \sum_{i=1}^D 2^{-\ell_i} \leq 1,
$$
and hence the classical Kraft-McMillan Inequality is also valid.

Conversely, suppose that $(U_k)_{k=1}^m$ is a sequence of 
quantum codes having a common sequence of length codewords 
$(\ket{\psi_i})_{i=1}^D$ and let $\ell_i \in \N \cup \{ 0 \}$ be the length of $\ket{\psi_i}$ for all 
$i=1, \ldots , D$. 
Assume that  Inequality~\eqref{kraft_inequality} is satisfied. Hence the classical Kraft-McMillan inequality is valid.
 So, by the converse of the classical Kraft-McMillan 
Theorem, there exists a set of bit strings 
$\{\psi_i^\prime\}_{i=1}^D$ such that there are exactly 
$d_\ell$ many bit strings with length $\ell$ for each 
$\ell \in \N$, and the sequence $(\psi_i')_{i=1}^D$ forms the image of a uniquely 
decodable (and in fact, instantaneous) classical code. 
Then, the sequence 
$(\ket{\psi_i^\prime})_{i=1}^D$ of the corresponding 
qubit strings satisfies 
$\ket{\psi_i^\prime} \in (\mathbb{C}^2)^{\otimes \ell_i}$
and every sequence of isometries $(U_k')_{k=1}^m$ 
such that each $U_k^\prime = \sum_{i=1}^D \ketbra{\psi_i^\prime}{e_i^k}$  (also called 
classical-quantum code as in \cite{Androulakis2019Sep}) for $k=1,\ldots,m$, forms a
uniquely decodable sequence of quantum codes
with a common sequence of jointly 
orthonormal length codewords.
\end{proof}

The last ingredient that we need for the proof of our main
result is the definition of concatenations of certain 
rank-1 operators on the Fock space. Notice that 
Remark~\ref{rmk:well-definition_of_concatenation}
ensures that the concatenation of rank-1 operators as 
introduced in the next definition is well defined.

\begin{Def} \label{def:concatenation_of_rank-1}
Let $(\ket{\psi_i})_{i=1}^D$ be a jointly orthonormal sequence of length
codewords in the Fock space. If $m \in \N$ and  
$I= (i_1, i_2, \ldots ,i_m) \in \{ 1, \ldots , D\}^m$, then 
let
$$
\ket{ \underset{i \in I}{\circ} \psi_i} := \ket{\psi_{i_1} \circ \psi_{i_2} \circ \cdots \circ\psi_{i_m}}.
$$
Let $m \in \N$,  
and $\ket{x}, \ket{y}$ be linear combinations of concatenations of $m$-many $\ket{\psi_i}$'s. Similarly,
let $n\in \N$ and $ \ket{z}, \ket{w}$ 
 be linear combinations of concatenations of $n$-many 
$\ket{\psi_i}$'s.
Thus, assume that 
$$
\ket{x} = \sum_i \mu_i \ket{\underset{i' \in I_{x,i}}{\circ} \psi_{i'}}, \quad
\ket{y}= \sum_j \mu_j \ket{\underset{j' \in I_{y,j}}{\circ} \psi_{j'}}, \quad
\ket{z}= \sum_k \mu_k \ket{\underset{k' \in I_{z,k}}{\circ} \psi_{k'}} \text{ and }
\ket{w}=\sum_\ell \mu_\ell \ket{\underset{\ell' \in I_{w,\ell}}{\circ} \psi_{\ell'}},
$$ 
where $\mu_i, \mu_j, \mu_k, \mu_\ell \in \C$,
$I_{x,i}, I_{y,j} \in \{ 1, \ldots , D\}^m$, and $I_{z,k},I_{w,\ell} \in \{ 1,\ldots ,D\}^{n}$  
for all $i,j,k,\ell$.
Then, we define the \textbf{concatenation of the rank-1 operators} $\ketbra{x}{y}$ and $\ketbra{z}{w}$ by
$$
\ketbra{x}{y} \circ \ketbra{z}{w} := \ketbra{x \circ z}{y \circ w}.
$$
\end{Def}

As a further clarification of the above definition, notice that under its setting, 
we have that if $I=(i_1,\ldots , i_m) \in \{1, \ldots , D\}^m$ and 
$J=(j_1,\ldots , j_n) \in \{ 1, \ldots , D\}^n$ for some $m,n \in \N$, 
then let $I\circ J \in \{1, \ldots , D\}^{m+n}$ be defined by $I \circ J =(i_1,\ldots , i_m,j_1, \ldots , j_n )$. 
Hence, notice that 
$$
\ket{x \circ z} = \sum_{i,k} \mu_i \mu_k \ket{\underset{m \in I_{x,i}\circ I_{z,k}}{\circ} \psi_m}
\text{ and }
\ket{y \circ w} = \sum_{j,\ell} \mu_j \mu_\ell \ket{\underset{n \in I_{y,j} \circ I_{w,\ell}}{\circ} \psi_n}. 
$$
Thus,
$$
\ketbra{x \circ z}{y \circ w} = \sum_{i,k,j,\ell} \mu_i \mu_k \overline{\mu_j \mu_\ell} 
\ketbra{\underset{m \in I_{x,i}\circ I_{z,k}}{\circ} \psi_m}{\underset{n \in I_{y,j} \circ I_{w,\ell}}{\circ} \psi_n}.
$$

Now we are ready to present the proof of our main result.

\begin{proof}[Proof of Theorem~\ref{Thm:main}] 
Let $\mathcal{U}= \{ U^{n_1, \ldots , n_{(k-1)l}} \}$ be any arbitrary
special block code, which is used to encode $m$ many blocks of pure states emitted by a given 
quantum stochastic source $\mathcal{S}$ (Definition~\ref{quantum_source}), where each block size is $l$. 
Let $(\ket{\psi_i})_{i=1}^{d^l} \subseteq (\mathbb{C}^2)^\oplus$ be 	
a sequence of jointly orthonormal length codewords which belongs in the range of every
member of $\mathcal{U}$ such that each $\ket{\psi_i}$ has 
length $\ell_i$ for some $\ell_i \in \N$ i.e.\ 
$\Tr(\ketbra{\psi_i} \Lambda) = \ell_i$, or equivalently 
$\ket{\psi_i} \in (\mathbb{C}^2)^{\otimes \ell_i}$. Let the orthonormal sequences be
$(\ket{e_i})_{i=1}^{d^l}$ for $k=1$ , and 
$(\ket{e_i^{n_1, \ldots , n_{(k-1)l}}})_{i=1}^{d^l}$ for 
$ 2 \leq k \leq m$ and 
$n_1, \ldots , n_{(k-1)l} \in \{ 1, \ldots , N \}$ such that 

\begin{equation} \label{E:definition_of_U}
U= \sum_{i=1}^{d^l} \ketbra{\psi_i}{e_i} 
\quad \text{and} \\ \quad
U^{n_1, \ldots , n_{(k-1)l}} = \sum_{i=1}^{d^l} \ketbra{\psi_i}{e_i^{n_1, \ldots , n_{(k-1)l}}} .
\end{equation}

Since we want to minimize the average codeword length, we encode the ensemble state $\rho_{ml}$ (Definition~\ref{general_ensemble}) for $ml$ emissions from the source $\mathcal{S}$ using the quantum codes of the special block code $\mathcal{U}$.
Recall that the average codeword length of $\mathcal{U}$ is defined in 
Definition~\ref{average_codeword} to be equal to
\begin{align*}
\sum_{n_1,\ldots,n_{ml}=1}^N p(n_1,\ldots,n_{ml})
& \Tr \Bigg( \ket{U (s_{n_1} \cdots s_{n_{l}}) \circ
		U^{n_1, \ldots , n_l} (s_{n_{l+1}} \cdots s_{n_{2l}}) \circ \cdots \circ U^{{n_1,\ldots , n_{(m-1)l}}} (s_{n_{(m-1)l+1}} \cdots s_{n_{ml}})} \\
&   \bra{
		U(s_{n_1} \cdots s_{n_{l}}) \circ 
		U^{n_1, \ldots , n_{n_l}} 
		(s_{n_{l+1}} \cdots s_{n_{2l}}) \circ \cdots 
		\circ U^{{n_1,\ldots , n_{(m-1)l}}} (s_{n_{(m-1)l+1}} \cdots s_{n_{ml}}) } \Lambda \Bigg)
\end{align*}

We will obtain the minimum average codeword length over all special block codes in two steps: in the first step, we optimize the orthonormal basis $\{\ket{e_i^{n_1, \ldots , n_{(k-1)l}}}\}_{i=1}^{d^l}$ for each isometry $U^{n_1, \ldots , n_{(k-1)l}} \in \mathcal{U}$, and in the second step, we optimize the common jointly orthonormal sequence of length codewords $(\ket{\psi_i})_{i=1}^{d^l}$. 

From now on, in the interest of minimizing the horizontal space used by equations, we will abbreviate 
certain scalars, pure states, mixed states and sums as 
in the following table 
for $r \leq t, i \in \N$. 


\begin{table}[H] 
	\centering
	\begin{tabular}{|l|l|l|} 
		\hline
		& The expression & is abbreviated by \\
		\hline
		\multirow{2}{1in}{Scalars:} & $p(n_r,\ldots,n_t)$ & $p[r,t]$ \\
		\cline{2-3}
		& $\lambda_i^{n_r, \ldots, n_t}$ & $\lambda_i^{[r,t]}$\\
		\hline
		\multirow{5}{1in}{Pure states:} & $\ket{s_{n_r} \cdots s_{n_t}}$ & $\ket{s_{[r,t]}}$ \\ 
		\cline{2-3}
		& $\ket{\lambda_i^{n_r, \ldots, n_t}}$ & $\ket{\lambda_i^{[r,t]}}$ \\ 
		\cline{2-3}
		& $\ket{e^{n_1, \ldots , n_t}_u}$   &  $\ket{e^t_u}$ \\ 
		\cline{2-3}
		&  $\ket{e_{u_0}e_{u_1}^{n_1,\ldots , n_l}e_{u_2}^{n_1,\ldots , n_{2l}} \cdots e_{u_t}^{n_1, \ldots , n_{tl}}}$  &  $\ket{e_{[u_0,u_t]}^{[0, tl]}}$ \\ 
		\cline{2-3}
		& $\ket{\psi_{u_r} \cdots \psi_{u_t}}$ &   $\ket{\psi_{[u_r,u_t]}}$  \\ 
		\hline

		Mixed states: & $\rho^{n_r, \ldots, n_t}$ & $\rho^{[r,t]}$ \\
		\hline
		\multirow{3}{1in}{Sums:} &
		$\sum_{n_r, \ldots, n_t =1}^N$ & $\sum_{[n_r,n_t]}^N$ \\
		\cline{2-3}
		& $\sum_{u_r, \ldots,u_t =1}^{d^l}$ &  $\sum_{[u_r,u_t]}^{d^l}$   \\
		\cline{2-3}
		& $\sum_{u_r, \ldots,u_t, u_r^{'}, \ldots, u_t^{'}=1}^{d^l}$ &  $\sum_{[u_r, u_t, u_r^{'},u_t^{'}]}^{d^l}$ \\
		\hline 
	\end{tabular}
	\caption{The list of all abbreviations that are used in this proof.\label{table_of_abbreviations}}
\end{table}
Also we accept the following convention.

\begin{Conv} \label{Conv:initial_terms}
	For $t<0$, the state
	$\ket{e_{u_0}e_{u_1}^{n_1,\ldots , n_l}e_{u_2}^{n_1,\ldots , n_{2l}} \cdots e_{u_t}^{n_1, \ldots , n_{tl}}}$ and its abbreviation $\ket{e_{[u_0,u_t]}^{[0, tl]}}$ do not exist
	and these terms should be ignored where they appear.
	For $t=0$, the state
	$\ket{e_{u_0}e_{u_1}^{n_1,\ldots , n_l}e_{u_2}^{n_1,\ldots , n_{2l}} \cdots e_{u_t}^{n_1, \ldots , n_{tl}}}$ and its abbreviation $\ket{e_{[u_0,u_t]}^{[0, tl]}}$ stand for the 
	state $\ket{e_{u_0}}$. 
\end{Conv}

Using these abbreviations, the average codeword length of $\mathcal{U}$ can be written as:
\begin{equation} \label{E:ACL_computation_1}
\begin{gathered}
	\sum_{[n_1,n_{ml}]}^N p[1, ml]
	\Tr \Big( \ket{U (s_{[1,l]}) \circ
		U^{n_1, \ldots , n_l} (s_{[l+1,2l]}) \circ \cdots \circ U^{{n_1,\ldots , n_{(m-1)l}}} (s_{[(m-1)l+1, ml]})} \\
	\bra{U(s_{[1,l]}) \circ 
		U^{n_1, \ldots , n_l} 
		(s_{[l+1,2l]}) \circ \cdots 
		\circ U^{{n_1,\ldots , n_{(m-1)l}}} (s_{[(m-1)l+1, ml]}) } \Lambda \Big).
\end{gathered}
\end{equation}

The expression inside the trace in the above formula of the average codeword length of $\mathcal{U}$ can be written as 
\begin{align*}
& \ket{U (s_{[1,l]})  \circ \cdots \circ U^{n_1,\ldots , n_{(m-2)l}}s_{[(m-2)l+1,(m-1)l]}} \circ
\ket{	U^{n_1,\ldots , n_{(m-1)l}} (s_{[(m-1)l+1, ml]})} \\
&  \bra{U (s_{[1,l]})  \circ \cdots \circ U^{n_1,\ldots , n_{(m-2)l}}s_{[(m-2)l+1,(m-1)l]}} \circ
\bra{	U^{n_1,\ldots , n_{(m-1)l}} (s_{[(m-1)l+1, ml]})}  \Lambda .
\end{align*}

We rewrite the last expression as
\begin{equation} \label{E:use_of_advanced_concatenations}
\begin{gathered}
\ketbra{U (s_{[1,l]})  \circ \cdots \circ U^{n_1,\ldots , n_{(m-2)l}} (s_{[(m-2)l+1,(m-1)l]})} \\
\circ
\ketbra{U^{n_1,\ldots , n_{(m-1)l}} (s_{[(m-1)l+1, ml]})}
\Lambda ,
\end{gathered}
\end{equation}
where the concatenation of the two rank-1 operators is defined via Definition~\ref{def:concatenation_of_rank-1}.
By \eqref{E:use_of_advanced_concatenations}, Expression~\eqref{E:ACL_computation_1} becomes:
\begin{align} \label{E:ACL_computation_2}
&\sum_{[n_1,n_{ml}]}^N p[1, ml] \nonumber\\
&\quad \quad \Tr \Big( \ketbra{U (s_{[1,l]})  \circ \cdots \circ U^{n_1,\ldots , n_{(m-2)l}}(s_{[(m-2)l+1,(m-1)l]})} \nonumber \\
&\quad \quad \circ
\ketbra{U^{n_1,\ldots , n_{(m-1)l}} (s_{[(m-1)l+1, ml]})} \Lambda \Big) \nonumber\\
&= \sum_{[n_1,n_{(m-1)l}]}^N p[1, (m-1)l] \nonumber\\
&\quad \quad \Tr \Bigg( \ketbra{U (s_{[1,l]})  \circ \cdots \circ U^{n_1,\ldots , n_{(m-2)l}}(s_{[(m-2)l+1,(m-1)l]})} \nonumber\\
& \quad \quad \circ \sum_{[n_{(m-1)l+1},n_{ml}]}^N \frac{p[1,ml]}{p[1,(m-1)l]} 
\ketbra{U^{n_1,\ldots , n_{(m-1)l}} (s_{[(m-1)l+1, ml]})} \Lambda \Bigg)\nonumber\\
&= \sum_{[n_1,n_{ml}]}^N p[1, (m-1)l]\nonumber\\
&\quad \quad \Tr \Bigg( \ketbra{U (s_{[1,l]})  \circ \cdots \circ U^{n_1,\ldots , n_{(m-2)l}}(s_{[(m-2)l+1,(m-1)l]})} \nonumber\\
& \quad \quad \circ U^{n_1,\ldots , n_{(m-1)l}} \Bigg( \sum_{[n_{(m-1)l+1},n_{ml}]}^N \frac{p[1,ml]}{p[1,(m-1)l]} 
\ketbra{ s_{[(m-1)l+1, ml]}} \Bigg) (U^{n_1,\ldots , n_{(m-1)l}})^\dagger \Lambda \Bigg)\nonumber\\
&= \sum_{[n_1,n_{ml}]}^N p[1, (m-1)l] \nonumber\\
&\quad \quad \Tr \Bigg( \ketbra{U (s_{[1,l]})  \circ \cdots \circ U^{n_1,\ldots , n_{(m-2)l}}(s_{[(m-2)l+1,(m-1)l]})} \nonumber\\
& \quad \quad \circ U^{n_1,\ldots , n_{(m-1)l}} \Big( \rho^{[1,(m-1)l]} \Big) (U^{n_1,\ldots , n_{(m-1)l}})^\dagger \Lambda \Bigg),
\end{align}
where $\rho^{[1,(m-1)l]}$ is the $m^{th}$ block conditional ensemble state as in Definition~\ref{block_ensemble_state}. 
If we use Equation~\eqref{E:definition_of_U}, the abbreviations listed in Table~\ref{table_of_abbreviations}, and the Convention~\ref{Conv:initial_terms},  
this expression becomes equal to the following:
\begin{equation} \label{E:ACL_computation_3}
\begin{gathered}
\sum_{[n_1,n_{(m-1)l}]}^N p[1,(m-1)l] \\
\Tr \Bigg(\sum_{[u_0,u_{m-2},u_0^{'}, u_{m-2}^{'}]}^{d^l} 
 \braket{e_{[u_0,u_{m-2}]}^{[0,(m-2)l]}}{s_{[1,{(m-1)l}]}} \braket{s_{[1,{(m-1)l}]}}{e_{[u_0^{'},u_{m-2}^{'}]}^{[0,(m-2)l]}}   \ketbra{\psi_{[u_0, u_{m-2}]}}{\psi_{[u_0^{'}, u_{m-2}^{'}]}} \\ 
\circ U^{{n_1,\ldots , n_{(m-1)l}}}
\rho^{[1, (m-1)l]}
(U^{{n_1,\ldots , n_{(m-1)l}}})^\dagger \ \Lambda \Bigg).
\end{gathered}
\end{equation}

Consider the spectral decomposition of the quantum state $\rho^{[1,(m-1)l]}$ given as
\begin{equation} \label{last_spectral}
	\rho^{[1,(m-1)l]} = \sum_{i=1}^{d^l} \lambda_{i}^{[1,(m-1)l]} \ketbra{\lambda_{i}^{[1,(m-1)l]}}{\lambda_{i}^{[1,(m-1)l]}}
\end{equation}
where each $\lambda_{i}^{[1,(m-1)l]}$ is an eigenvalue corresponding to the eigenvector $\ket{\lambda_{i}^{[1,(m-1)l]}}$, (see Table~\ref{table_of_abbreviations}). 
Without loss of generality, we assume that $\lambda_{1}^{[1,(m-1)l]} \geq \lambda_{2}^{[1,(m-1)l]} \cdots \geq \lambda_{d^l}^{[1,(m-1)l]} \geq 0$.
Thus, using Equation~\eqref{last_spectral} we obtain that 
the average codeword length of $\mathcal{U}$ (i.e.\ Expression~\eqref{E:ACL_computation_3})is equal to

\begin{align} 
& \sum_{[n_1,n_{(m-1)l}]}^N p[1,(m-1)l] \nonumber \\
& \quad \quad \Tr \Bigg( \sum_{[u_0,u_{m-2},u_0^{'}, u_{m-2}^{'}]}^{d^l} 
\braket{e_{[u_0,u_{m-2}]}^{[0,(m-2)l]}}{s_{[1,{(m-1)l}]}} \braket{s_{[1,{(m-1)l}]}}{e_{[u_0^{'},u_{m-2}^{'}]}^{[0,(m-2)l]}}   \ketbra{\psi_{[u_0, u_{m-2}]}}{\psi_{[u_0^{'}, u_{m-2}^{'}]}} \nonumber \\ 
& \quad \quad 
 \sum_{i = 1}^{d^l} \lambda_{i}^{[1,(m-1)l]} \sum_{[u_{m-1},u_{m-1}^{'}]}^{d^l}\braket{e_{u_{m-1}}^{m-1}}{\lambda_{i}^{[1,(m-1)l]}} \braket{\lambda_{i}^{[1,(m-1)l]}}{e_{u_{m-1}^{'}}^{m-1}} \ketbra{\psi_{u_{m-1}}}{\psi_{u_{m-1}^{'}}}   \Lambda 
\Bigg) \nonumber \\
&= \sum_{[n_1,n_{(m-1)l}]}^N p[1,(m-1)l] 
 \sum_{[u_0,u_{m-2},u_0^{'}, u_{m-2}^{'}]}^{d^l}  \braket{e_{[u_0,u_{m-2}]}^{[0,(m-2)l]}}{s_{[1,{(m-1)l}]}} \braket{s_{[1,{(m-1)l}]}}{e_{[u_0^{'},u_{m-2}^{'}]}^{[0,(m-2)l]}}  \nonumber \\  
& \quad \quad
\sum_{i = 1}^{d^l} \lambda_i^{[1,(m-1)l]} \sum_{[u_{m-1},u_{m-1}^{'}]}^{d^l}\braket{e_{u_{m-1}}^{m-1}}{\lambda_{i}^{[1,(m-1)l]}} \braket{\lambda_{i}^{[1,(m-1)l]}}{e_{u_{m-1}^{'}}^{m-1}}
\Tr \Bigg(\ketbra{\psi_{[u_0,u_{m-1}]}}{\psi_{[u_0^{'},u_{m-1}^{'}]}} \Lambda \Bigg)  \nonumber \\ 
&= \sum_{[n_1,n_{(m-1)l}]}^N p[1,(m-1)l] 
\sum_{[u_0,u_{m-2},u_0^{'}, u_{m-2}^{'}]}^{d^l}  \braket{e_{[u_0,u_{m-2}]}^{[0,(m-2)l]}}{s_{[1,{(m-1)l}]}} \braket{s_{[1,{(m-1)l}]}}{e_{[u_0^{'},u_{m-2}^{'}]}^{[0,(m-2)l]}}  \nonumber \\  
& \quad \quad
\sum_{i = 1}^{d^l} \lambda_{i}^{[1,(m-1)l]} \sum_{[u_{m-1},u_{m-1}^{'}]}^{d^l}\braket{e_{u_{m-1}}^{m-1}}{\lambda_{i}^{[1,(m-1)l]}} \braket{\lambda_{i}^{[1,(m-1)l]}}{e_{u_{m-1}^{'}}^{m-1}} \sum_{\ell=0}^\infty \ell \matrixel{\psi_{[u_0^{'},u_{m-1}^{'}]}}{\Pi_\ell}{\psi_{[u_0,u_{m-1}]}}    \label{collapse_um},
\end{align}
where in the last equality we used the cyclic property of trace and Equation~\eqref{length_observ}. 

Then, we can evaluate  part 
of the Equation~\eqref{collapse_um} as follows:
$$
\sum_{\ell=0}^\infty \ell \matrixel{\psi_{[u_0^{'},u_{m-1}^{'}]}}{\Pi_\ell}{\psi_{[u_0,u_{m-1}]}} 
= \left( \sum_{j=0}^{m-1} \ell_{u_j} \right) 
\left( \prod_{k=0}^{m-1} \delta_{u_k,u_k^{'}} \right)
$$ 
since $(\ket{\psi_{i}})_{i=1}^{d^l}$ is a jointly orthonormal 
sequence of length codewords (Definition~\ref{jointly_orthonormal}). So, 
Equation~\eqref{collapse_um} simplifies to the following expression:
\begin{equation} \label{bistochastic}
\sum_{[n_1,n_{(m-1)l}]}^N p[1,(m-1)l] \sum_{[u_0,u_{m-2}]}^{d^l} \left( \left| \braket{s_{[1,{(m-1)l}]}}{e_{[u_0,u_{m-2}]}^{[0,(m-2)l]}}
\right|^2  
 \sum_{i= 1}^{d^l} \lambda_{i}^{[1,(m-1)l]} \sum_{u_{m-1}=1}^{d^l}\left| \braket{\lambda_{i}^{[1,(m-1)l]}} {e_{u_{m-1}}^{m-1}}\right| ^2  \sum_{j=0}^{m-1} \ell_{u_j} \right) .
\end{equation}
Notice that since 
$\left\{ \ket{\lambda_i^{[1,(m-1)l]}} \right\}_{i=1}^{d^l}$
and $\left\{ \ket{e_{u_{m-1}}^{m-1}} \right\}_{u_{m-1}=1}^{d^l}$ are both  orthonormal bases of $\mathcal{H}^{\otimes l}$, 
 we have
$$
\sum_{i=1}^{d^l} \left| \braket{\lambda_i^{[1,(m-1)l]}} {e_{u_{m-1}}^{m-1}} \right|^2 
=  \sum_{u_{m-1}=1}^{d^l} \left| \braket{\lambda_i^{[1,(m-1)l]}}{e_{u_{m-1}}^{m-1}}\right|^2  
=  1.
$$
Therefore, for every $n_1, \ldots, n_{(m-1)l} \in 
\{ 1, \ldots , N \}$, the terms 
\begin{equation} \label{E:entries_of_B}
\left\{ 
\left| \braket{\lambda_i^{[1,(m-1)l]}}{e_{u_{m-1}}^{m-1}}\right|^2 
: i,u_{m-1}= 1, \ldots , d^l \right\}
\end{equation}
in Expression~\eqref{bistochastic} above, form  a 
$d^l \times d^l$ bistochastic matrix $B^{n_1, \ldots, n_{(m-1)l}}$. 
By the Birkoff-von~Neumann Theorem, every $d^l \times d^l$ 
bistochastic matrix can be written as a convex combination 
of at most $d^l!$ many permuation matrices. So,
\begin{equation}\label{b_matrix}
B^{n_1, \ldots, n_{(m-1)l}} = \sum_{h=1}^{d^l!} t_h^{n_1, \ldots, n_{(m-1)l}} P_h^{n_1, \ldots, n_{(m-1)l}}, 
\end{equation}
where $ \left\{ t_h^{n_1, \ldots, n_{(m-1)l}} \right\}_{h=1}^{d^l!}$ 
are  convex coefficients and 
$\left\{ P_h^{n_1, \ldots, n_{(m-1)l}} 
\right\}_{h=1}^{d^l!}$ are permutation matrices.
For every fixed 
$n_1,\ldots , n_{(m-1)l} \in \{ 1, \ldots , N \}$, define the 
function $f^{n_1, \ldots , n_{(m-1)l}}$ on the
$d^l\times d^l$ matrices $B=(B(i,u_{m-1}))_{i,u_{m-1}=1}^{d^l}$ and taking scalar values, as follows
$$
f^{n_1, \ldots , n_{(m-1)l}}(B) := \sum_{[u_0,u_{m-2}]}^{d^l}  \left| \braket{s_{[1,{(m-1)l}]}}{e_{[u_0,u_{m-2}]}^{[0,(m-2)l]}}
\right|^2  
\sum_{i= 1}^{d^l} \lambda_{i}^{[1,(m-1)l]} \sum_{u_{m-1}=1}^{d^l} B(i,u_{m-1}) \sum_{j=0}^{m-1} \ell_{u_j}. 
$$
Notice that the function $f^{n_1, \ldots , n_{(m-1)l}}$ is affine. Thus, for fixed  $n_1,\ldots , n_{(m-1)l} \in \{ 1, \ldots , N \}$, 
there exists $h \in \{ 1, \ldots , d^l! \}$, (which depends 
on $n_1,\ldots , n_{(m-1)l}$), such that the minimizer 
$B^{n_1,\ldots ,n_{(m-1)l}}$ of $f^{n_1, \ldots , n_{(m-1)l}}$ 
is equal to $P_h^{n_1,\ldots , n_{(m-1)l}}$, i.e.

\begin{align*}
& f^{n_1, \ldots , n_{(m-1)l}}(B^{n_1,\ldots ,n_{(m-1)l}}) \\
\geq &
f^{n_1, \ldots , n_{(m-1)l}}(P_h^{n_1,\ldots , n_{(m-1)l}}) \\
=& \sum_{[u_0,u_{m-2}]}^{d^l} \left| \braket{s_{[1,{(m-1)l}]}}{e_{[u_0,u_{m-2}]}^{[0,(m-2)l]}}\right|^2   
 \sum_{i= 1}^{d^l} \lambda_{i}^{[1,(m-1)l]} \sum_{u_{m-1}=1}^{d^l}   P_h^{n_1, \ldots, n_{(m-1)l}}(i,u_{m-1})  \sum_{j=0}^{m-1} \ell_{u_j} ,
\end{align*}
\noindent
where $P_h^{n_1, \ldots, n_{(m-1)l}}(i,u_{m-1})$ denotes the 
$(i, u_{m-1})$ entry of $P_h^{n_1, \ldots, n_{(m-1)l}}$.
Since $P_h^{n_1, \ldots, n_{(m-1)l}}$ is a $d^l\times d^l$ 
permutation matrix, for every $1\leq i \leq d^l$ there exists
a unique $g_i \in \{ 1, \ldots , d^l\}$ such that 
$P_h^{n_1, \ldots, n_{(m-1)l}}(i, g_i) = 1$, and all other 
entries of the $i^{\text{th}}$ row are equal to $0$. Hence,

\begin{equation} \label{E:entries_of_P}
P_h^{n_1, \ldots, n_{(m-1)l}} (i, u_{m-1})= \delta_{u_{m-1},g_i}.
\end{equation}

Therefore,

\begin{align}
& \sum_{[u_0,u_{m-2}]}^{d^l} \left| \braket{s_{[1,{(m-1)l}]}}{e_{[u_0,u_{m-2}]}^{[0,(m-2)l]}}\right|^2   
\sum_{i= 1}^{d^l} \lambda_{i}^{[1,(m-1)l]} \sum_{u_{m-1}=1}^{d^l}   P_h^{n_1, \ldots, n_{(m-1)l}}(i,u_{m-1})  \sum_{j=0}^{m-1} \ell_{u_j} \nonumber \\
=& \sum_{[u_0,u_{m-2}]}^{d^l} \left| \braket{s_{[1,{(m-1)l}]}}{e_{[u_0,u_{m-2}]}^{[0,(m-2)l]}}\right|^2   \sum_{i=1}^{d^l} \lambda_{i}^{[1,(m-1)l]} \Bigg(\sum_{j=0}^{m-2} \ell_{u_j} + \ell_{g_i} \Bigg). \label{lemma_use}
\end{align}

For the length sequence $(\ell_i)_{i=1}^{d^l}$ of the jointly orthonormal sequence of length codewords $(\ket{\psi})_{i=1}^{d^l}$, we can assume without loss of generality that 
\begin{equation} \label{E:increasing_l_s}
\ell_1 \leq \ell_2 \leq \cdots \leq \ell_{d^l}.
\end{equation}

Since we have also assumed that 
$\lambda_{1}^{[1,(m-1)l]} \geq \lambda_{2}^{[1,(m-1)l]} \geq  \cdots \geq \lambda_{d^l}^{[1,(m-1)l]} \geq 0$, we 
obtain by 
Lemma~\ref{lemma}, that Equation ~\eqref{lemma_use} is 
minimized when 

\begin{equation} \label{E:g_i=i}
g_i = i \text{ for }1 \leq i \leq d^l. 
\end{equation}

Equations~\eqref{E:entries_of_P} and \eqref{E:g_i=i} imply that the  
minimizer permutation matrix $P^{n_1, \ldots, n_{(m-1)l}}$ is
equal to the identity  $d^l \times d^l$ matrix. 
Since the entries of the input matrices of the function $f^{n_1,\ldots , n_{(m-1)l}}$ are given by \eqref{E:entries_of_B}, we obtain that 
$$
\left| \braket{\lambda_i^{[1,(m-1)l]}}{e_{u_{m-1}}^{m-1}}\right|
= \delta_{i, u_{m-1}} \quad \text{for }i,u_{m-1}= 1, \ldots , d^l.
$$ 
Since both vectors $\ket{\lambda_i^{[1,(m-1)l]}}$ and $\ket{e_{u_{m-1}}^{m-1}}$ are normalized, one can infer that 
$$
\ket{e_i^{m-1}} = \alpha_i \ket{\lambda_i^{[1,(m-1)l]}} \text{ for some } \alpha_i \in \mathbb{C} \text{ with } |\alpha_i| =1 \text{ for }1 \leq i \leq d^l.
$$
By ignoring the phase factors $\alpha_i$'s, (since they 
do not 
affect the average codeword length of $\mathcal{U}$, as 
it can be seen 
by the Expression~\eqref{bistochastic}), we conclude that 
the average codeword length of $\mathcal{U}$ will be minimized if we use the following isometry to encode the last ($m^\text{th}$) block:
\begin{equation} \label{E:U_m}
U^{n_1,\ldots , n_{(m-1)l}} = \sum_{i=1}^{d^l} \ketbra{\psi_i}{\lambda_i^{[1,(m-1)l]}}.
\end{equation}
Equation~\eqref{E:g_i=i} also gives that the minimum of  Expression~\eqref{bistochastic} is equal to
\begin{equation} \label{E:end_of_mth_step}
	\sum_{[n_1,n_{(m-1)l}]}^N p[1,(m-1)l] \sum_{[u_0,u_{m-2}]}^{d^l} \left| \braket{s_{[1,{(m-1)l}]}}{e_{[u_0,u_{m-2}]}^{[0,(m-2)l]}}\right|^2   \sum_{i=1}^{d^l} \lambda_{i}^{[1,(m-1)l]} \Bigg(\sum_{j=0}^{m-2} \ell_{u_j} + \ell_i \Bigg).
\end{equation}
One can proceed with finding the optimal orthonormal basis 
for $(m-1)^{th}$ block in a similar manner, as we explain now. 
Equation~\eqref{E:end_of_mth_step} can be split into two sums 
as follows:

\begin{equation} \label{E:begining_of_m-1_block_part_1}
\begin{gathered}
\sum_{[n_1,n_{(m-1)l}]}^N p[1,(m-1)l] \sum_{[u_0,u_{m-2}]}^{d^l} \left| \braket{s_{[1,{(m-1)l}]}}{e_{[u_0,u_{m-2}]}^{[0,(m-2)l]}}\right|^2 \sum_{i=1}^{d^l} \lambda_i^{[1,(m-1)l]} \ell_i \ + \\ 
\sum_{[n_1,n_{(m-1)l}]}^N p[1,(m-1)l] \sum_{[u_0,u_{m-2}]}^{d^l}  \left| \braket{s_{[1,{(m-1)l}]}}{e_{[u_0,u_{m-2}]}^{[0,(m-2)l]}}\right|^2 \sum_{i=1}^{d^l} \lambda_i^{[1,(m-1)l]}\left( \sum_{j=0}^{m-2} \ell_{u_j}\right)  .
\end{gathered}
\end{equation}

Note that 

$$
\sum_{[u_0,u_{m-2}]}^{d^l} \left| \braket{s_{[1,{(m-1)l}]}}{e_{[u_0,u_{m-2}]}^{[0,(m-2)]}}\right|^2  = \norm{\ket{s_{[1,{(m-1)l}]}}}^2=1 \quad \text{and} \quad \sum_{i=1}^{d^l} \lambda_i^{[1,(m-1)l]} = 1,
$$

hence Expression~\eqref{E:begining_of_m-1_block_part_1} simplifies as 
follows:

\begin{equation} \label{E:begining_of_m-1_block_part_2}
\begin{gathered}
\sum_{[n_1,n_{(m-1)l}]}^N p[1,(m-1)l] \sum_{i= 1}^{d^l} \lambda_{i}^{[1,(m-1)l]} \ell_{i} \  + \\ 
\sum_{[n_1,n_{(m-1)l}]}^N p[1,(m-1)l] \sum_{[u_0,u_{m-2}]}^{d^l} \left| \braket{s_{[1,{(m-1)l}]}}{e_{[u_0,u_{m-2}]}^{[0,(m-2)l]}}\right|^2 \left( \sum_{j=0}^{m-2} \ell_{u_j}\right) .   
\end{gathered}
\end{equation}
It is worth mentioning at this point that splitting Equation~\eqref{E:end_of_mth_step} followed by some simplifications gave Equation~\eqref{E:begining_of_m-1_block_part_2} which has two parts: a) the first line i.e. $\sum_{[n_1,n_{(m-1)l}]}^N p[1,(m-1)l] \sum_{i= 1}^{d^l} \lambda_{i}^{[1,(m-1)l]} \ell_{i}$ represents the minimum average codeword length for the $m^{th}$ (last) block for the fixed set of lengths $(\ell_i)_{i=1}^{d^l}$ and b) the second line represents the sum of the average codeword lengths for the first $(m-1)$ blocks which still needs to be minimized. As we will see below, this second line can be similarly split further into the minimum average codeword length for the $(m-1)^{th}$ block for the given set of lengths $(\ell_i)_{i=1}^{d^l}$ and the sum of the average codeword lengths for the first $(m-2)$ blocks subject to further minimization. Therefore, the minimum average codeword length for $m$ blocks is the sum of the minimum average codeword lengths for each block.\\

Notice that the second line of Expression~\eqref{E:begining_of_m-1_block_part_2}
can be simplified as follows:
\begin{align*}
&\sum_{[u_0,u_{m-2}]}^{d^l} \left| \braket{s_{[1,{(m-1)l}]}}{e_{[u_0,u_{m-2}]}^{[0,(m-2)l]}}\right|^2 \left( \sum_{j=0}^{m-2} \ell_{u_j}\right) \\
=& \sum_{[u_0,u_{m-3}]}^{d^l} \left| \braket{s_{[1,{(m-2)l}]}}{e_{[u_0,u_{m-3}]}^{[0,(m-3)l]}}\right|^2 \sum_{u_{m-2}=1}^{d^l} \left| \braket{s_{[{(m-2)l+1},{(m-1)l]}}}{e_{u_{m-2}}^{(m-2)l}}\right|^2 \left( \sum_{j=0}^{m-3} \ell_{u_j} + \ell_{u_{m-2}} \right)    \\
= & \sum_{[u_0,u_{m-3}]}^{d^l} \left| \braket{s_{[1,{(m-2)l}]}}{e_{[u_0,u_{m-3}]}^{[0,(m-3)l]}}\right|^2 \sum_{u_{m-2}=1}^{d^l} \left| \braket{s_{[{(m-2)l+1},{(m-1)l]}}}{e_{u_{m-2}}^{(m-2)l}}\right|^2 \left( \sum_{j=0}^{m-3} \ell_{u_j} \right) +\\
& \sum_{[u_0,u_{m-3}]}^{d^l} \left| \braket{s_{[1,{(m-2)l}]}}{e_{[u_0,u_{m-3}]}^{[0,(m-3)l]}}\right|^2 \left( \sum_{u_{m-2}=1}^{d^l} \left| \braket{s_{[{(m-2)l+1},{(m-1)l]}}}{e_{u_{m-2}}^{(m-2)l}}\right|^2 \ell_{u_{m-2}} \right) \\
= & \sum_{[u_0,u_{m-3}]}^{d^l} \left| \braket{s_{[1,{(m-2)l}]}}{e_{[u_0,u_{m-3}]}^{[0,(m-3)l]}}\right|^2 \norm{\ket{s_{[{(m-2)l+1},{(m-1)l]}}}} ^2\left( \sum_{j=0}^{m-3} \ell_{u_j} \right) +\\
& \sum_{[u_0,u_{m-3}]}^{d^l} \left| \braket{s_{[1,{(m-2)l}]}}{e_{[u_0,u_{m-3}]}^{[0,(m-3)l]}}\right|^2 \left( \sum_{u_{m-2}=1}^{d^l}  \braket{e_{u_{m-2}}^{(m-2)l}}{s_{[{(m-2)l+1},{(m-1)l]}}} \braket{s_{[{(m-2)l+1},{(m-1)l]}}}{e_{u_{m-2}}^{(m-2)l}} \ell_{u_{m-2}} \right) \\
= & \sum_{[u_0,u_{m-3}]}^{d^l} \left| \braket{s_{[1,{(m-2)l}]}}{e_{[u_0,u_{m-3}]}^{[0,(m-3)l]}}\right|^2 \left( \sum_{j=0}^{m-3} \ell_{u_j} \right) +\\
& \sum_{[u_0,u_{m-3}]}^{d^l} \left| \braket{s_{[1,{(m-2)l}]}}{e_{[u_0,u_{m-3}]}^{[0,(m-3)l]}}\right|^2 \left( \sum_{u_{m-2}=1}^{d^l} \ell_{u_{m-2}}  \braket{e_{u_{m-2}}^{(m-2)l}}{s_{[{(m-2)l+1},{(m-1)l]}}}
\braket{s_{[{(m-2)l+1},{(m-1)l]}}}{e_{u_{m-2}}^{(m-2)l}} \right). 
\end{align*}

By substituting the last quantity in Expression~\eqref{E:begining_of_m-1_block_part_2} we obtain

\begin{equation} \label{E:begining_of_m-1_block_part_2.5}
\begin{gathered}
\sum_{[n_1,n_{(m-1)l}]}^N p[1,(m-1)l] \sum_{i= 1}^{d^l} \lambda_{i}^{[1,(m-1)l]} \ell_{i} \  +  \\ 
\sum_{[n_1,n_{(m-1)l}]}^N p[1,(m-1)l]  \sum_{[u_0,u_{m-3}]}^{d^l} \left| \braket{s_{[1,{(m-2)l}]}}{e_{[u_0,u_{m-3}]}^{[0,(m-3)l]}}\right|^2 \left( \sum_{j=0}^{m-3} \ell_{u_j} \right) +  \\
\sum_{[n_1,n_{(m-1)l}]}^N p[1,(m-1)l] \sum_{[u_0,u_{m-3}]}^{d^l} \left| \braket{s_{[1,{(m-2)l}]}}{e_{[u_0,u_{m-3}]}^{[0,(m-3)l]}}\right|^2  \\
\left(\sum_{u_{m-2}=1}^{d^l} \ell_{u_{m-2}}  \braket{e_{u_{m-2}}^{(m-2)l}}{s_{[{(m-2)l+1},{(m-1)l]}}}\braket{s_{[{(m-2)l+1},{(m-1)l]}}}{e_{u_{m-2}}^{(m-2)l}} \right) .
\end{gathered}
\end{equation}

Note that since the bras $\bra{s_{n_{(m-2)l+1}} \cdots s_{n_{(m-1)l}}}$ do not appear any more in 
the second line of Expression~\eqref{E:begining_of_m-1_block_part_2.5}, the probability distributions $p[1,(m-1)l]$ 
in the second line of the Expression~\eqref{E:begining_of_m-1_block_part_2.5} when summed over the indices 
$n_{(m-2)l+1}, \ldots , n_{(m-1)l}$ simplify to $p[1,(m-2)l]$.
Thus, Expression~\eqref{E:begining_of_m-1_block_part_2.5} simplifies to

\begin{align} \label{E:begining_of_m-1_block_part_3}
& \sum_{[n_1,n_{(m-1)l}]}^N p[1,(m-1)l] \sum_{i= 1}^{d^l} \lambda_{i}^{[1,(m-1)l]} \ell_{i} \  + \nonumber \\ 
& \sum_{[n_1,n_{(m-2)l}]}^N p[1,(m-2)l]  \sum_{[u_0,u_{m-3}]}^{d^l} \left| \braket{s_{[1,{(m-2)l}]}}{e_{[u_0,u_{m-3}]}^{[0,(m-3)l]}}\right|^2 \left( \sum_{j=0}^{m-3} \ell_{u_j} \right) + \nonumber \\
& \sum_{[n_1,n_{(m-1)l}]}^N p[1,(m-1)l] \sum_{[u_0,u_{m-3}]}^{d^l} \left| \braket{s_{[1,{(m-2)l}]}}{e_{[u_0,u_{m-3}]}^{[0,(m-3)l]}}\right|^2 \nonumber \\
& \left(\sum_{u_{m-2}=1}^{d^l} \ell_{u_{m-2}}  \braket{e_{u_{m-2}}^{(m-2)l}}{s_{[{(m-2)l+1},{(m-1)l]}}}\braket{s_{[{(m-2)l+1},{(m-1)l]}}}{e_{u_{m-2}}^{(m-2)l}} \right) \nonumber \\
=& \sum_{[n_1,n_{(m-1)l}]}^N p[1,(m-1)l] \sum_{i= 1}^{d^l} \lambda_{i}^{[1,(m-1)l]} \ell_{i} \  + \nonumber \\ 
& \sum_{[n_1,n_{(m-2)l}]}^N p[1,(m-2)l]  \sum_{[u_0,u_{m-3}]}^{d^l} \left| \braket{s_{[1,{(m-2)l}]}}{e_{[u_0,u_{m-3}]}^{[0,(m-3)l]}}\right|^2 \left( \sum_{j=0}^{m-3} \ell_{u_j} \right) + \nonumber \\
& \sum_{[n_1,n_{(m-2)l}]}^N p[1,(m-2)l] \sum_{[u_0,u_{m-3}]}^{d^l} \left| \braket{s_{[1,{(m-2)l}]}}{e_{[u_0,u_{m-3}]}^{[0,(m-3)l]}}\right|^2 \nonumber \\
& \left(\sum_{u_{m-2}=1}^{d^l} \ell_{u_{m-2}}  \matrixel{e_{u_{m-2}}^{(m-2)l}}{\left( \sum_{[n_{(m-2)l+1},n_{(m-1)l}]}^N \frac{p[1,(m-1)l]}{p[1,(m-2)l]} \ketbra{s_{[{(m-2)l+1},{(m-1)l]}}}{s_{[{(m-2)l+1},{(m-1)l]}}}\right) }{e_{u_{m-2}}^{(m-2)l}} \right) \nonumber \\
=& \sum_{[n_1,n_{(m-1)l}]}^N p[1,(m-1)l] \sum_{i= 1}^{d^l} \lambda_{i}^{[1,(m-1)l]} \ell_{i} \  + \nonumber \\ 
& \sum_{[n_1,n_{(m-2)l}]}^N p[1,(m-2)l]  \sum_{[u_0,u_{m-3}]}^{d^l} \left| \braket{s_{[1,{(m-2)l}]}}{e_{[u_0,u_{m-3}]}^{[0,(m-3)l]}}\right|^2 \left( \sum_{j=0}^{m-3} \ell_{u_j} \right) + \nonumber \\
& \sum_{[n_1,n_{(m-2)l}]}^N p[1,(m-2)l] \sum_{[u_0,u_{m-3}]}^{d^l} \left| \braket{s_{[1,{(m-2)l}]}}{e_{[u_0,u_{m-3}]}^{[0,(m-3)l]}}\right|^2 \nonumber \\
& \left(\sum_{u_{m-2}=1}^{d^l} \ell_{u_{m-2}}  \matrixel{e_{u_{m-2}}^{(m-2)l}}{ \rho^{[1,(m-2)l]}}{e_{u_{m-2}}^{(m-2)l}} \right).
\end{align}

As in Equation~\eqref{last_spectral}, consider the spectral decomposition of $\rho^{[1,(m-2)l]}$,
\begin{equation} \label{secondlast_spectral}
\rho^{[1,(m-2)l]} = \sum_{i=1}^{d^l} \lambda_{i}^{[1,(m-2)l]} \ketbra{\lambda_{i}^{[1,(m-2)l]}}{\lambda_{i}^{[1,(m-2)l]}}.
\end{equation}
By substituting Equation~\eqref{secondlast_spectral} in Equation~\eqref{E:begining_of_m-1_block_part_3} we obtain:

\begin{equation} \label{E:begining_of_m-1_block_part_4}
\begin{gathered}
\sum_{[n_1,n_{(m-1)l}]}^N p[1,(m-1)l] \sum_{i= 1}^{d^l} \lambda_{i}^{[1,(m-1)l]} \ell_{i} \  +  \\ 
 \sum_{[n_1,n_{(m-2)l}]}^N p[1,(m-2)l]  \sum_{[u_0,u_{m-3}]}^{d^l} \left| \braket{s_{[1,{(m-2)l}]}}{e_{[u_0,u_{m-3}]}^{[0,(m-3)l]}}\right|^2 \left( \sum_{j=0}^{m-3} \ell_{u_j} \right) +  \\
 \sum_{[n_1,n_{(m-2)l}]}^N p[1,(m-2)l] \sum_{[u_0,u_{m-3}]}^{d^l} 
 \left| \braket{s_{[1,{(m-2)l}]}}{e_{[u_0,u_{m-3}]}^{[0,(m-3)l]}}\right|^2  \\
\left( \sum_{i=1}^{d^l}  \lambda_{i}^{[1,(m-2)l]} \sum_{u_{m-2}=1}^{d^l} \ell_{u_{m-2}}   \left| \braket{\lambda_{i}^{[1,(m-2)l]}}{e_{u_{m-2}}^{(m-2)l}} \right|^2   \right) .
\end{gathered}
\end{equation}

Similarly as before, (see Expression~\eqref{bistochastic}), for every $n_1, \ldots, n_{(m-2)l}$, 
each $\left| \braket{\lambda_i^{[1,(m-2)l]}}{e_{u_{m-2}}^{(m-2)}}\right|^2$ in Equation~\eqref{E:begining_of_m-1_block_part_4} 
above forms the $(i,u_{m-2})$ element of a $d^l \times d^l$ bistochastic matrix $B^{n_1, \ldots, n_{(m-2)l}}$ 
for $1\leq i, u_{m-2} \leq d^l$. 
By the Birkoff-von~Neumann Theorem and Lemma~\ref{lemma}, Equation~\eqref{E:begining_of_m-1_block_part_4} attains its minimum at the 
 $d^l \times d^l$ identity permutation matrix. 
Thus, 
$$
\ket{e_i^{m-2}} = \alpha_i \ket{\lambda_i^{[1,(m-2)l]}} \text{ for some } \alpha_i \in \mathbb{C} \text{ with } |\alpha_i| =1 \text{ for }1 \leq i \leq d^l,
$$

and hence, by ignoring the phase factors $\alpha_i$'s as before,  we conclude that
the average codeword length of $\mathcal{U}$ will be minimized if we use the following isometry to encode the last $(m-1)^\text{th}$ block:

\begin{equation} \label{E:U_m-1}
U^{n_1,\ldots , n_{(m-2)l}} = \sum_{i=1}^{d^l} \ketbra{\psi_i}{\lambda_i^{[1,(m-2)l]}}.
\end{equation} 

In this case,
Equation~\eqref{E:begining_of_m-1_block_part_4} simplifies to

\begin{align} \label{E:begining_of_m-1_block_part_5}
&\sum_{[n_1,n_{(m-1)l}]}^N p[1,(m-1)l] \sum_{i= 1}^{d^l} \lambda_{i}^{[1,(m-1)l]} \ell_{i} \  +  \nonumber \\ 
& \sum_{[n_1,n_{(m-2)l}]}^N p[1,(m-2)l]  \sum_{[u_0,u_{m-3}]}^{d^l} \left| \braket{s_{[1,{(m-2)l}]}}{e_{[u_0,u_{m-3}]}^{[0,(m-3)l]}}\right|^2 \left( \sum_{j=0}^{m-3} \ell_{u_j} \right) +  \nonumber \\
& \sum_{[n_1,n_{(m-2)l}]}^N p[1,(m-2)l] \sum_{[u_0,u_{m-3}]}^{d^l} 
\left| \braket{s_{[1,{(m-2)l}]}}{e_{[u_0,u_{m-3}]}^{[0,(m-3)l]}}\right|^2  
\left( \sum_{i=1}^{d^l}  \lambda_{i}^{[1,(m-2)l]}  \ell_i    \right) \nonumber \\
= & \sum_{[n_1,n_{(m-1)l}]}^N p[1,(m-1)l] \sum_{i= 1}^{d^l} \lambda_{i}^{[1,(m-1)l]} \ell_{i} \  +  \nonumber \\ 
& \sum_{[n_1,n_{(m-2)l}]}^N p[1,(m-2)l]  \sum_{[u_0,u_{m-3}]}^{d^l} \left| \braket{s_{[1,{(m-2)l}]}}{e_{[u_0,u_{m-3}]}^{[0,(m-3)l]}}\right|^2 \left( \sum_{i=1}^{d^l}  \lambda_{i}^{[1,(m-2)l]}  \ell_i    +\sum_{j=0}^{m-3} \ell_{u_j} \right) .
\end{align}

By comparing the Expressions~\eqref{E:end_of_mth_step} and \eqref{E:begining_of_m-1_block_part_5}, one can easily guess the 
formula of the minimum average codeword length of $\mathcal{U}$ obtained by optimizing the orthonormal basis for each isometry in the family $\mathcal{U}$. We chose to write the details of the first two steps rather than 
the formal induction step, since it is easier to be followed by the reader. The minimum average length of $\mathcal{U}$ is
equal to

\begin{align} \label{final}
&\sum_{[n_1,n_{(m-1)l}]}^N p[1,(m-1)l] \sum_{i= 1}^{d^l} \lambda_{i}^{[1,(m-1)l]} \ell_{i} +
\sum_{[n_1,n_{(m-2)l}]}^N p[1,(m-2)l] \sum_{i= 1}^{d^l} \lambda_{i}^{[1,(m-2)l]} \ell_{i} + \nonumber \\
& \sum_{[n_1,n_{(m-3)l}]}^N p[1,(m-3)l] \sum_{i= 1}^{d^l} \lambda_{i}^{[1,(m-3)l]} \ell_{i} +  \cdots + 
\sum_{n_1=1}^N p[1,l] \sum_{i= 1}^{d^l} \lambda_{i}^{[1,l]} \ell_{i} + 
\sum_{i=1}^{d^l} \lambda_i \ell_i \nonumber \\
= & \sum_{j=1}^{m-1}\left(\sum_{[n_1,n_{(m-j)l}]}^N p[1,(m-j)l] \sum_{i= 1}^{d^l} \lambda_{i}^{[1,(m-j)l]} \ell_{i} \right) +
\sum_{i=1}^{d^l} \lambda_i \ell_i,
\end{align}
\noindent
where $\{\lambda_i\}_{i=1}^{d^l}$ are the eigenvalues of the ensemble state $\rho_l$, i.e.
$$
\rho_l = \sum_{[n_1,n_l]}^{d^l} p[1,l]\ketbra{s_{[1,l]}}  
=\sum_{i=1}^{d^l} \lambda_i \ketbra{\lambda_i} ,
$$
where $\lambda_i$ is the eigenvalue corresponding to the 
eigenvector $\ket{\lambda_i}$.

Similarly to Equations~\eqref{E:U_m} and \eqref{E:U_m-1}, we 
obtain that the average codeword length of $\mathcal{U}$ will be 
minimized if we use the following isometries to encode each of 
the blocks:

\begin{equation} \label{E:U_m-j}
U^{n_1,\ldots , n_{(m-j)l}} = \sum_{i=1}^{d^l} \ketbra{\psi_i}{\lambda_i^{[1,(m-j)l]}}, \text{ for }
j=1, \ldots , m-1, \text{ and } U = \sum_{i=1}^{d^l} \ketbra{\psi_i}{\lambda_i}.
\end{equation}

Equation~\eqref{final} can be re-written as 
\begin{equation} \label{re-write_final}
\sum_{j=2}^{m}\left(\sum_{[n_1,n_{(j-1)l}]}^N p[1,(j-1)l] \sum_{i= 1}^{d^l} \lambda_{i}^{[1,(j-1)l]} \ell_{i} \right) +
\sum_{i=1}^{d^l} \lambda_i \ell_i,
\end{equation}

Now we move to the optimization of the sequence of jointly orthonormal length codewords $(\ket{\psi_i})_{i=1}^{d^l}$.
By the 
forward direction of quantum Kraft-McMillan inequality (Theorem~\ref{generalized_Kraft}), any finite sequence of quantum codes from the special 
block code is uniquely decodable and the length sequence $(\ell_i)_{i=1}^{d^l}$ of the jointly orthonormal sequence of length codewords $(\ket{\psi})_{i=1}^{d^l}$ must satisfy the classical Kraft-McMillan inequality 
\begin{equation} \label{classical_Kraft}
	\sum_{i=1}^{d^l} 2^{-\ell_i} \leq 1.
\end{equation}
Also recall that  without loss of generality, the 
Inequality~\eqref{E:increasing_l_s} is satisfied. 
Thus, in order to minimize the average codeword length of $\mathcal{U}$, 
we need to choose the optimal length sequence $(\ell_i)_{i=1}^{d^l}$, 
which in addition to satisfying 
Equations~\eqref{classical_Kraft} and\eqref{E:increasing_l_s},
also minimizes Equation~\eqref{re-write_final}. Notice that a complete binary Kraft tree of height $\lceil \log_2 {d^l} \rceil$ has at least $d^l$ many leaf nodes which are enough to assign uniquely decodable codewords to any set of $d^l$ many symbols. Therefore, for any optimal length 
sequence $(\ell_i)_{i=1}^{d^l}$, we have 
$1 \leq \ell_i \leq \ell_{max} = \lceil \log_2 {d^l} \rceil$ for every $i$. So, it is clear that one needs to test finitely many such 
sequences of lengths $( \ell_i )_{i=1}^{d^l}$ to find the 
optimal length sequence, and therefore
the minimizer that we seek is well defined.

Once the sequence $(\ell_i)_{i=1}^{d^l}$ has been chosen as described in the above paragraph, then 
choose a classical binary sequence $(\omega_i)_{i=1}^{d^l}$ using the binary Kraft tree as shown in 
\cite[Page 107, Theorem 5.2.1]{Cover2005Apr} such that the length of 
the codeword $\omega_i$ (i.e.\ the number of its binary digits), is equal to $\ell_i$. Now, choose 
the sequence $(\ket{\psi_i})_{i=1}^{d^l}$ to be equal to the sequence of qubit strings $(\ket{\omega_i})_{i=1}^{d^l}$.
It is worth mentioning that the codewords $(\omega_i)_{i=1}^{d^l}$ obtained from the construction given in \cite[Page 107, Theorem 5.2.1]{Cover2005Apr} are classical prefix-free codes. However, it is well known by the classical Kraft-McMillan inequality (and its converse) that for every uniquely decodable codeword sequence with the lengths $(\ell_i)_{i=1}^{d^l}$, there also exists a prefix-free codeword sequence with the same lengths. So, it suffices to construct the sequence of jointly orthonormal length codewords $(\ket{\psi_i})_{i=1}^{d^l}$ from the classical prefix-free codes without affecting the optimality.
Additionally, since $(\omega_i)_{i=1}^{d^l}$ is prefix-free (and hence, uniquely decodable), their concatenations denoted by $\omega_{i_1} \circ \cdots \circ \omega_{i_m}$ for $i_1, \ldots, i_m \in \{1, \ldots, d^l\}$ are pairwise distinct for any $m \in \N$. Since each $\ket{\omega_i} \in (\mathbb{C}^2)^{\otimes \ell_i}$, the indeterminate length of $\ket{\omega_i}$ is equal to $\Tr(\Lambda \ketbra{\omega_i}) = \ell_i$. Hence, each $\ket{\omega_i}$ is a length codeword. Therefore, their concatenations $\ket{\omega_{i_1} \circ \cdots \circ \omega_{i_m}}$ of length $m$ are always well defined and orthonormal for every $m \in \N$. So, the sequence $(\ket{\omega_i})_{i=1}^{d^l}$ forms a sequence of jointly orthonormal length codewords. Lastly, notice that this sequence $(\ket{\omega_i})_{i=1}^{d^l}$ is not unique as one can find several such jointly orthonormal sequences of length codewords each having the same optimal length sequence.
Finally, 
the optimal special block code is then given by
$$
\mathcal{V}= \Big\{ V^{n_1, \ldots , n_{(k-1)l}}: 
k \in \{ 1, \ldots , m \},\text{ and }
n_1, \ldots ,n_{(k-1)l} \in \{ 1, \ldots , N \} \Big\},
$$ 
(where for $k=1$, the string $n_1, \ldots , n_{(k-1)l}$
is the empty string, and $V^{n_1, \ldots , n_{(k-1)l}}$ 
is simply denoted by $V$), defined by
$$
V^{n_1, \ldots , n_{(k-1)l}} = \sum_{i=1}^{d^l} \ketbra{\omega_i}{\lambda_i^{n_1, \ldots, n_{(k-1)l}}},
$$
(where for $k=1$, the vectors $(\ket{\lambda_i^{n_1, \ldots , n_{(k-1)l}}})_{i=1}^{d^l}$ stand for the eigenvectors
$(\ket{\lambda_i})_{i=1}^{d^l}$ of the ensemble state
$\rho_l$). 

Since the length sequence $(\ell_i)_{i=1}^{d^l}$ of the obtained sequence of jointly orthonormal length codewords $(\ket{\omega_i})_{i=1}^{d^l}$ minimizes Equation~\ref{re-write_final}, it follows that Equation~(\ref{E:optimality}) is satisfied and therefore, $L(\mathcal{V})$ is the lower bound on the average codeword length over all special block codes.
\end{proof}

\section{Proofs of Theorem~\ref{Thm:derived} and Corollary~\ref{corollary_stationary}} \label{additional_theorems}
The proof of Theorem~\ref{Thm:derived} proceeds similarly as that of Theorem~\ref{Thm:main}.
\begin{proof}[Sketch of the proof of Theorem~\ref{Thm:derived}]
As before, fix $m, l$, and an arbitrary $k^{th}$ block conditional ensemble state $\rho^{n_1, \ldots, n_{(k-1)l}} = \sum_{i=1}^{d^l} \lambda_i^{n_1, \ldots, n_{(k-1)l}} \ketbra{\lambda_i^{n_1, \ldots, n_{(k-1)l}}}$ such that $\lambda_r^{n_1, \ldots, n_{(k-1)}} \geq \lambda_s^{n_1, \ldots, n_{(k-1)l}}$ for $r \leq s$. Define $\mathfrak{L}$ as before. Fix a constrained special block code $\mathcal{U} = \{U_k\}_{k=1}^m$ such that each isometry $U_k =  \sum_{i=1}^{d^l} \ketbra{\psi_i}{e_i^k} \in \mathcal{U}$ has $(\ket{\psi_i})_{i=1}^{d^l}$ as common jointly orthonormal sequence of length codewords and an arbitrary orthonormal sequence $(\ket{e_i^k})_{i=1}^{d^l}$. For every $k^{th}$ block, Alice uses the isometry $U_k$ to encode the symbols. Suppose that length$(\ket{\psi_i}) = \ell_i \in \N $ and without loss of generality, assume $\ell_r \leq \ell_s$ for $r \leq s$. By the forward direction of quantum Kraft-McMillan inequality (Theorem~\ref{generalized_Kraft}), we have that the set of non-negative integer lengths $\{\ell_i\}_{i=1}^{d^l}$ belongs to $\mathfrak{L}$.

In the sketch of the proof of Theorem~\ref{Thm:main}, by using the Birkhoff-von Neumann Theorem and Lemma \ref{lemma}, we argued that for a given set of lengths $(\ell_i)_{i=1}^{d^l} \in \mathfrak{L}$, in order to achieve the absolute minimization of the average codeword length, every $\rho^{n_1, \ldots, n_{(k-1)l}}$ needs to be encoded using its corresponding optimal isometry $U^{n_1, \ldots, n_{(k-1)l}} = \sum_{i=1}^{d^l} \ketbra{\psi_i}{\lambda_i^{n_1, \ldots, n_{(k-1)l}}}$. However, Alice now favors simplicity at the cost of higher average codeword length, and uses a fixed isometry $U_k = \sum_{i=1}^{d^l} \ketbra{\psi_i}{e_i^k}$ to encode each $\rho^{n_1, \ldots, n_{(k-1)l}}$ for every sequence of states $\ket{s_{n_1}, \ldots, s_{n_{(k-1)l}}}$ as the first $(k-1)l$ emissions. So, we can assert that 
\begin{equation} \label{compare_two_lengths}
\Tr (U_k \rho^{n_1, \ldots, n_{(k-1)l}} U_k^\dagger \Lambda) \geq \Tr (U^{n_1, \ldots, n_{(k-1)l}} \rho^{n_1, \ldots, n_{(k-1)l}} (U^{n_1, \ldots, n_{(k-1)l}})^\dagger \Lambda)
\end{equation}
Now the average codeword length for encoding the $k^{th}$ block ensemble state 
$$
\rho^k = \sum_{n_1, \ldots, n_{(k-1)l} =1}^N p(n_1, \ldots, n_{(k-1)l}) \rho^{n_1, \ldots, n_{(k-1)l}}
$$ is given by 
\begin{equation} \label{k-th block length}
\sum_{n_1, \ldots, n_{(k-1)l}} p(n_1, \ldots, n_{(k-1)l}) \Tr (U_k \rho^{n_1, \ldots, n_{(k-1)l}} U_k^\dagger \Lambda)
\end{equation}
for the fixed set of lengths $\{\ell_i\}_{i=1}^{d^l}$. Since $U_k$ is independent of $\ket{s_{n_1}, \ldots, s_{n_{(k-1)l}}}$, using the linearity of trace and the definition of $\rho^k$ (Definition \ref{block_ensemble_state}), Equation~\eqref{k-th block length} can be re-written as 
\begin{align} \label{optimized k-th block length}
 \Tr (U_k \left( \sum_{n_1, \ldots, n_{(k-1)l}} p(n_1, \ldots, n_{(k-1)l}) \rho^{n_1, \ldots, n_{(k-1)l}} \right )U_k^\dagger \Lambda) = \Tr (U_k  \rho^k U_k^\dagger \Lambda).
\end{align}
Consider the spectral decomposition of the $k^{th}$ block ensemble state $\rho^k$ as $\rho^k = \sum_{i=1}^{d^l} \lambda_i^k \ketbra{\lambda_i^k}$ such that $\lambda_r^k \geq \lambda_s^k$ for $r \leq s$. Then, it can be shown again by using the Birkhoff-von Neumann Theorem and Lemma \ref{lemma} that $\Tr (U_k  \rho^k U_k^\dagger \Lambda)$ is minimized for the given set of lengths $\{\ell_i\}_{i=1}^{d^l}$ when $U_k = \sum_{i=1}^{d^l} \ketbra{\psi_i}{\lambda_i^k}$. If we use $U_k = \sum_{i=1}^{d^l} \ketbra{\psi_i}{\lambda_i^k}$, then  $\Tr (U_k  \rho^k U_k^\dagger \Lambda)$ simplifies to $\sum_{i=1}^{d^l} \lambda_i^k \ell_i$. Summing such length for each $\rho^k$ gives Equation~\eqref{LC}, which is the minimum average codeword length of encoding $m$ blocks for a fixed set of lengths $\{\ell_i\}_{i=1}^{d^l}$. Refer to the paragraph following Equation~\eqref{E:begining_of_m-1_block_part_2} in the Appendix \ref{main_theorem} for the details. Minimizing Equation~\eqref{LC} over 
all sets of lengths in $\mathfrak{L}$ gives Equation~\eqref{E:Best_ACL_C}, which is the minimum average codeword length of encoding $m$ blocks over the set of all constrained special block codes. Once the minimizer set of lengths is found, one can then construct an optimal jointly orthonormal sequence of length codewords $(\ket{\omega_i})_{i=1}^{d^l}$ and define a minimizer special block code $\mathcal{V}$ as stated in the theorem. 
\end{proof}

From Equation~(\ref{compare_two_lengths}), it should now be clear that for the same set of lengths $\{\ell_i\}_{i=1}^{d^l}$, the minimum average codeword length produced by this non-adaptive compression for encoding $\rho^k$ is at least as large as that produced by the adaptive compression for every $k^{th}$ block. Precisely, this means that for the same set of lengths $\{\ell_i\}_{i=1}^{d^l}$,
$$
\sum_{i=1}^{d^l} \lambda_i^k \ell_i \geq \sum_{n_1, \ldots, n_{(k-1)l}} p(n_1, \ldots, n_{(k-1)l}) \sum_{i=1}^{d^l} \lambda_i^{n_1, \ldots, n_{(k-1)l}} \ell_i
$$ for every $k^{th}$ block. Consequently, it's not hard to see that the minimum average codeword length produced by the non-adaptive compression over the set of all constrained special block codes for encoding $m$ blocks of pure states is at least as large as that produced by the adaptive compression over the set of all special block codes. \\

The proof of Corollary~\ref{corollary_stationary} proceeds as follows:
\begin{proof}[Proof of Corollary~\ref{corollary_stationary}]
If the quantum source is stationary, then we can show as below that all the $k^{th}$ block ensemble states $\rho^k$'s are identical.
Indeed, from Equations~\eqref{kth_block_ensemble} and \eqref{block_ensemble}, we have 
\begin{align*}
\rho^k = & \sum_{n_1, \ldots, n_{(k-1)l} = 1}^N \quad \sum_{n_{(k-1)l+1}, \ldots, n_{kl} =1}^N p(n_1, \ldots, n_{kl}) \ketbra{s_{n_{(k-1)l+1}} \cdots s_{n_{kl}}} \\
& = \sum_{n_{(k-1)l+1}, \ldots, n_{kl} =1}^N  \ketbra{s_{n_{(k-1)l+1}} \cdots s_{n_{kl}}}  \sum_{n_1, \ldots, n_{(k-1)l} = 1}^N p(n_1, \ldots, n_{kl}) \\
& = \sum_{n_{(k-1)l+1}, \ldots, n_{kl} =1}^N  \ketbra{s_{n_{(k-1)l+1}} \cdots s_{n_{kl}}} \mathbb{P}(X_{(k-1)l+1} = s_{n_{(k-1)l+1}}, \ldots, X_{kl} = s_{n_{kl}}) \text{}\\
& =  \sum_{n_{(k-1)l+1}, \ldots, n_{kl} =1}^N  \ketbra{s_{n_{(k-1)l+1}} \cdots s_{n_{kl}}} \mathbb{P}(X_1 = s_{n_{(k-1)l+1}}, \ldots, X_l = s_{n_{kl}})  \\ 
&(\text{from Equation~\eqref{translation-invariance}}) \\
& =  \sum_{n_1, \ldots, n_l =1}^N  \ketbra{s_{n_1} \cdots s_{n_l}} \mathbb{P}(X_1 = s_{n_1}, \ldots, X_l = s_{n_l}) \quad \text{(equivalently re-writing above equation)} \\
& =  \sum_{n_1, \ldots, n_l =1}^N  \ketbra{s_{n_1} \cdots s_{n_l}} p(n_1, \ldots, n_l) \quad \text{(from Equation~\eqref{stochastic_px}}) \\
& = \rho_l \quad \text{(from Definition~\ref{general_ensemble})}.
\end{align*}
Therefore, we have $\rho^1 = \cdots = \rho^k = \rho_{l}$. Consider the spectral decomposition of $\rho_{l}$ and the isometry $V = \sum_{i=1}^{d^l} \ketbra{\omega_i}{\lambda_i}$ as in the statement of the corollary. This isometry $V$ is the optimal isometry \cite[Theorem 2]{Bellomo2017Nov} for encoding the state $\rho_{l}$. As a result, it turns out that in order to optimally encode every $\rho^k$, Alice can use the same isometry $V$. In other words, the minimizer constrained special block code for a quantum stationary source is $\mathcal{V} = V$. Therefore, the minimum average codeword length for encoding $m$ blocks of states will be $ILC(\mathcal{S}, m, l) = m\Tr (V  \rho_{l} V^\dagger \Lambda) = m\sum_{i=1}^{d^l} \lambda_i \ell_i$. Then the minimum average codeword length per block is $\frac{ILC(\mathcal{S}, m, l)}{m}=\Tr (V  \rho_{l} V^\dagger \Lambda) = \sum_{i=1}^{d^l} \lambda_i \ell_i = ILC(\mathcal{S}, 1, l)$. Notice that the minimum average codeword length per block depends only on the block size $l$ and is independent of the number of blocks $m$. We can now directly apply \cite[Theorem 3]{Bellomo2017Nov} to bound this minimum average codeword length per block from above and below in terms of the von-Neumann entropy of the source as follows:
\begin{equation} \label{stationary_bound}
S(\rho_{l}) \leq \sum_{i=1}^{d^l} \lambda_i \ell_i \leq S(\rho_{l}) + 1,
\end{equation}
where $S(\rho_{l})$ represents the von-Neumann entropy of $\rho_{l}$.
Dividing the above Inequality (\ref{stationary_bound}) throughout by the block size $l$, one obtains
\begin{equation} \label{stationary_bound_rate}
\frac{S(\rho_{l})}{l} \leq \frac{1}{l} \sum_{i=1}^{d^l} \lambda_i \ell_i \leq \frac{S(\rho_{l})}{l} + \frac{1}{l},
\end{equation}
where $\frac{1}{l} \sum_{i=1}^{d^l} \lambda_i \ell_i = \frac{ILC(\mathcal{S}, 1, l)}{l}$ represents the minimum average codeword length per symbol over all constrained special block codes. Finally, taking the limit $l \to \infty$ in the above Inequality~\ref{stationary_bound_rate}, we get
\begin{equation}
\lim_{l \to \infty} \ \frac{ILC(\mathcal{S}, 1, l)}{l} = \frac{1}{l} \sum_{i=1}^{d^l} \lambda_i \ell_i = \lim_{l \to \infty} \frac{S(\rho_{l})}{l}
\end{equation}
where $ \lim_{l \to \infty} \frac{S(\rho_{l})}{l}$ represents the von-Neumann entropy rate of the stationary source which is known to always exist \cite[Theorem 9.14]{alicki2001quantum}.
\end{proof}

\end{document}